\documentclass[9pt, journal,twoside,web]{ieeecolor}
\usepackage{generic}
\usepackage[noadjust]{cite}
\usepackage{amsmath,amssymb,amsfonts}
\usepackage{graphicx}
\usepackage{textcomp}
\usepackage{enumerate}
\usepackage{mathrsfs}
\usepackage{color}
\usepackage[justification=centering]{caption}
\usepackage{balance}
\usepackage{algorithm}
\usepackage{algorithmic}
\usepackage{verbatim}
\usepackage{romannum}
\usepackage{subfig}
\let\labelindent\relax
\usepackage{enumitem}
\allowdisplaybreaks[4]
\makeatletter
\let\NAT@parse\undefined
\makeatother
\usepackage[colorlinks,
            linkcolor=black,
            anchorcolor=blue,
            citecolor=black
            ]{hyperref}
\newtheorem{remark}{Remark}
\newtheorem{assumption}{Assumption}
\newtheorem{lemma}{Lemma}
\newtheorem{theorem}{Theorem}
\newtheorem{problem}{Problem}
\newtheorem{definition}{Definition}
\newtheorem{corollary}{Corollary}
\newtheorem{proposition}{Proposition}
\newcommand{\IEEEQED}{\hfill\rule{2mm}{2mm}}

\definecolor{blue}{RGB}{0,0,190}

\def\BibTeX{{\rmB\kern-.05em{\sci\kern-.025emb}\kern-.08emT\kern-.1667em\lower.7ex\hbox{E}\kern-.125emX}}
\begin{document}
\pagenumbering{arabic}

\title{\huge To What Extent Can Inherent Communication Noise Guarantee Privacy in Distributed Cooperative Control?}
\author{Yuwen Ma, \IEEEmembership{Student Member, IEEE}, Sarah Spurgeon, \IEEEmembership{Fellow, IEEE}, Tao Li, \IEEEmembership{Senior Member, IEEE,} Boli Chen, \IEEEmembership{Senior Member, IEEE}
\thanks{This work was supported by Engineering and Physical Sciences
Research Council (EPSRC) of UK Research and Innovation (UKRI)
under Grant EP/W524335/1.}
\thanks{Yuwen Ma, Boli Chen and Sarah Spurgeon are with the Department of Electronic and Electrical Engineering, University College London, London, United Kingdom (e-mail: yuwen.ma.24@ucl.ac.uk, boli.chen@ucl.ac.uk, s.spurgeon@ucl.ac.uk).}
\thanks{Tao Li is with the Key Laboratory of Management, Decision and Information Systems, Institute of Systems Science, Academy of Mathematics
and Systems Science, Chinese Academy of Sciences, Beijing 100190,
China, and also with the School of Mathematical Sciences, University of Chinese Academy of Sciences, Beijing 100149, China (e-mail: litao@amss.ac.cn).}
}
\maketitle

\begin{abstract}
This paper proposes a differentially private distributed cooperative control scheme for multi-agent systems (MAS). Unlike conventional approaches that actively inject artificial noise for privacy protection, this work investigates whether inherent communication noise can itself serve as a natural privacy mechanism. A physically motivated communication-noise model is developed for mobile MAS by incorporating transmitter perturbation, receiver noise, path-loss attenuation, and log-normal shadowing. The resulting effective noise variance depends on inter-agent state differences, thereby capturing the distance-dependent signal perturbation arising in practice.
Based on this model, a distributed finite-horizon Linear Quadratic Regulator (LQR) mechanism is designed to achieve formation tracking while protecting agents’ private control preferences. Rather than protecting the full local cost function, the proposed privacy formulation focuses on the ratio of the LQR weighting matrices, which captures the trade-off between tracking accuracy and control effort when the quadratic cost structure is publicly known. A set-theoretic sensitivity analysis shows that this weighting-ratio adjacency formulation yields less conservative privacy bounds than gradient-based protection under the considered addition/removal adjacency relation.
Theoretical analysis demonstrates that, under suitable design conditions, the proposed mechanism provides bounded cumulative $(\epsilon,\delta)$-differential privacy guarantees for the weighting ratios over an infinite horizon without artificial noise injection. Meanwhile, the cooperative tracking error is shown to converge almost surely and in mean square to a finite random limit, with its expectation remaining bounded. Numerical examples validate the theoretical results and illustrate the resulting privacy–performance trade-off.
\end{abstract}


\section{Introduction}
\label{sec:introduction}
\IEEEPARstart {D}{} {i}stributed cooperative control of multi-agent systems has attracted sustained research attention due to its broad applications in mobile robotics, vehicular platooning, unmanned aerial vehicle (UAV) swarms, and other networked autonomous systems \cite{cao2025proximal,ma2025reduced,oh2015survey,RenCommunication}. A central objective in such systems is to coordinate a group of agents to achieve a prescribed collective behaviour, such as consensus, formation, or cooperative regulation, through local information exchange among neighbouring agents. However, the same communication process that enables cooperation may also expose sensitive information. For example, an external eavesdropper, or a compromised agent within the network, may exploit exchanged signals to infer private reference states, topologies, cost-function parameters, or individual system dynamics \cite{li2024preserving,yan2024privacy}. Therefore, privacy preservation has become an important issue in distributed cooperative control.

A range of privacy-preserving techniques has been studied for networked control and optimisation. Cryptographic approaches, such as encryption and secure multi-party computation, provide strong security guarantees, but typically rely on exact algebraic operations and may be difficult to integrate with stochastic control algorithms and real-time distributed implementation {\cite{darup2021encrypted}}. Differential privacy (DP) \cite{dwork2014algorithmic}, in contrast, provides a probabilistic privacy framework that quantifies the indistinguishability of algorithmic outputs under neighbouring private datasets. Its robustness to post-processing and side information makes it particularly attractive for distributed control and optimisation systems in which adversaries may observe communication signals over time.


Over the past decade, differentially private distributed algorithms have been developed for a variety of protection targets, including initial states \cite{nozari2017differentially,he2020differential}, local cost functions \cite{huang2015differentially,nozari2016differentially,cao2020differentially,huang2024differential,wang2024robust}, constraint functions, and feasible sets \cite{han2016differentially,zhang2024differentially}. Most existing approaches achieve DP by deliberately injecting artificial noise into communicated variables, gradients, objective functions, or control updates \cite{nozari2017differentially,he2020differential,huang2015differentially,nozari2016differentially,cao2020differentially,huang2024differential,wang2024robust,han2016differentially,zhang2024differentially,yazdani2022differentially,cortes2016differential,nozari2015differentially,hawkins2020differentially}. Although effective, such artificial perturbations inevitably introduce a trade-off between privacy and control performance. In particular, the amount of injected noise typically needs to be calibrated to the sensitivity of the protected quantity, and excessive perturbation may degrade convergence accuracy or closed-loop performance. This limitation raises a natural question: can randomness that already exists in the physical communication process be exploited as a privacy-preserving resourc?

In practical mobile multi-agent systems, inter-agent communication is rarely noise-free. Wireless channels are affected by transmitter distortion, receiver noise, path loss, fading, and obstacle-induced shadowing. Moreover, because the agents are mobile, the quality of the received signal often depends on the relative distance between agents. This means that communication noise is not merely an external disturbance to be suppressed; it may also provide a natural source of randomness for privacy preservation. Similar ideas have recently been explored in wireless federated learning, where channel noise has been used to reduce or even eliminate the need for additional artificial noise \cite{liu2020privacy,liu2022wireless,liu2024differentially}. In control and optimisation, inherent system imperfections such as quantisation errors have also been studied as privacy-preserving resources \cite{wang2022quantization,liu2026design}. However, the use of inherent communication noise as a direct DP mechanism for distributed cooperative control remains much less understood.

This paper is motivated by the following question: \textit{to what extent can inherent communication noise itself guarantee privacy in distributed cooperative control, without actively injecting additional noise?} Addressing this question is challenging because inherent communication noise is fundamentally different from artificial DP noise. Artificial noise can be designed and scaled according to a prescribed privacy budget, whereas physical channel noise is determined by the communication environment and is not freely tunable. Consequently, the achievable privacy level depends not only on the channel-noise intensity, but also on the sensitivity of the protected quantity and the associated adjacency relation.

This observation motivates a careful choice of the privacy target. Existing DP mechanisms for distributed optimisation often protect the full local objective function or its gradient \cite{huang2015differentially,nozari2016differentially,cao2020differentially,huang2024differential,wang2024robust,han2016differentially,zhang2024differentially}. While such formulations are appropriate in many optimisation problems, they may be overly conservative in control settings such as Linear Quadratic Regulator (LQR) and Model Predictive Control (MPC), where the quadratic structure of the cost function is typically known a priori \cite{di2018dynamic}. In such cases, the truly private information is not necessarily the mathematical form of the cost, but rather the preference encoded by the relative weighting of tracking accuracy and control effort. Motivated by this perspective, this paper focuses on protecting the weighting ratios of the quadratic LQR cost function. This formulation follows the spirit of the authors’ preliminary work on differentially private finite-horizon optimal control \cite{ma2025distributed}, while moving from artificially designed privacy noise to inherent communication noise. 

The present study also builds on the authors’ preliminary investigation of communication-noise-induced privacy in distributed cooperative control \cite{ma2026inherent}. Compared with that initial study, the present work considers a more realistic channel-noise model for mobile agents and develops a distributed finite-horizon LQR framework in which the protected private information is the weighting ratio of the quadratic cost function. In this way, the paper connects two previously separate ideas: using inherent communication noise as a privacy resource, and protecting control preferences through a weighting-ratio adjacency formulation. Specifically, this work investigates cooperative formation control of multi-agent systems in which the communication noise is scaled by inter-agent state differences, such as relative position gaps. A distributed finite-horizon LQR mechanism is proposed to achieve both formation tracking and privacy protection for the weighting ratios of the local quadratic cost functions. Theoretical analysis shows that, over an infinite horizon, the proposed mechanism can provide bounded DP guarantees while ensuring that the tracking error converges almost surely (a.s.) and in mean square to a finite random limit, with its expectation remaining strictly bounded. Numerical simulations are provided to validate the theoretical findings.

The main contributions of the paper are summarised as follows.

\begin{enumerate}

\item \textbf{A physically motivated communication-noise model for mobile multi-agent systems.}
    A distance-dependent channel-noise model is developed for mobile agents by incorporating transmitter perturbation, receiver noise, path-loss attenuation, and log-normal shadowing. The resulting effective communication noise depends on the inter-agent state difference, thereby capturing both signal perturbation and distance-dependent channel variation in cooperative mobile systems.

\item \textbf{A weighting-ratio adjacency formulation for LQR privacy.}
    Motivated by the fact that the quadratic LQR structure is often publicly known, the paper formulates the private database in terms of the weighting ratios $R_i^{-1}Q_i$, rather than the full cost function. This ratio-based formulation captures the agents’ control preferences and is shown, via set-theoretic sensitivity analysis, to be less conservative than gradient-based protection under the considered addition/removal adjacency formulation.


\item \textbf{A distributed finite-horizon LQR mechanism with inherent-noise-induced privacy guarantees.}
    A distributed receding-horizon LQR mechanism is proposed to achieve formation tracking while using inherent communication noise, rather than artificially injected noise, to protect the weighting ratios. The analysis establishes tracking-error convergence and bounded cumulative $(\epsilon,\delta)$-DP over an infinite horizon under suitable conditions, and it further clarifies the resulting privacy–performance trade-off.
\end{enumerate}
The remainder of this paper is organised as follows. Section~\uppercase\expandafter{\romannumeral2} gives the problem formulation and preliminaries, including the multi-agent system dynamics, the communication-noise model, the finite-horizon LQR-based cooperative control formulation, the definition of differential privacy, and the formal problem statement. Section~\uppercase\expandafter{\romannumeral3} develops the differentially private distributed LQR cooperative mechanism, and the convergence analysis of the tracking error, the DP guarantees induced by inherent communication noise, and the sensitivity-based justification for protecting weighting ratios rather than gradients. Section~\uppercase\expandafter{\romannumeral4} provides numerical simulations to validate the theoretical results and illustrate the privacy-performance trade-off. Finally, Section~\uppercase\expandafter{\romannumeral5} concludes the paper.

\textit{Notations:} $\mathbb{R}^{m\times n}$ and $\mathbb{R}^{n}$ denote the spaces of real $m\times n$ matrices and $n$-dimensional real vectors, respectively. The sets of positive definite and positive semi-definite matrices are represented by $\mathbb{S}^{n\times n}_{++}$ and $\mathbb{S}^{n\times n}_{+}$. The set of natural numbers is denoted by $\mathbb{N}$, with $\mathbb{N}_+$ referring specifically to the set of positive integers. Regarding matrix notation, $\mathbf{I}$ and $\mathbf{0}$ represent the identity and zero matrices of compatible dimensions, respectively, and $\mathbf{1}$ denotes the column vector with all entries equal to one. The notation $\mathrm{diag}\left\{ {{A}_{1}},\ldots, {{A}_{n}}\right\}$ corresponds to a diagonal matrix with diagonal entries ${{A}_{1}},\ldots,{{A}_{n}}$, and $\mathrm{range}(\cdot)$ denotes the range of an operator. The $l_\infty$-norm is indicated by ${\left\| \cdot \right\|_\infty}$, and unless stated otherwise, $\|\cdot\|$ refers to the standard Euclidean (${l}_2$) norm. For a vector $x \in \mathbb{R}^n$ and a matrix $P \in \mathbb{S}_{++}^{n \times n}$, $\|x\|_P^2 \triangleq x^\top P x$ denotes the squared weighted Euclidean norm. For $p \in \{1, 2\}$, the space $\ell_p$ is defined as the set of real sequences for which the $p$-th power of the absolute value is summable, i.e., $\ell_p = \{ x=(x_n)_{n=1}^{\infty} \mid \sum_{n=1}^{\infty} |x_n|^p < \infty \}$. Additionally, $\otimes$ indicates the Kronecker product. For a square matrix $A\in \mathbb{R}^{n\times n}$, $\mathrm{tr}(A)$ is its trace and $\lambda_{\min}(A)$ represents its smallest positive eigenvalue; for symmetric matrices, the inequalities $A \succ 0$ and $A \succeq 0$ signify positive definiteness and positive semi-definiteness, respectively. For two sets $A$ and $B$, $\rm{d_H}(A,B)$ denotes the Hausdorff distance between $A$ and $B$. In probabilistic terms, $\mathbb{E}[\cdot]$ denotes the mathematical expectation, $\mathrm{Cov[\cdot]}$ represents the covariance matrix, $\mathbb{P}(B)$ is the probability of an event $B$, and $\mathcal{N}(0,\,\Sigma)$ represents a zero-mean Gaussian distribution with covariance matrix $\Sigma$. For two real-valued sequences $\{a(t)\}$ and $\{b(t)\}$, the notation $a(t) = O(b(t))$ implies the existence of constants $c > 0$ and $t_0$ such that $|a(t)| \le c |b(t)|$ holds for all $t \ge t_0$.

\section{Problem formulation and preliminaries} 
\subsection{Multi-agent Systems}
A multi-agent system (MAS) consisting of $N$ agents is studied, where the dynamics of each agent are modelled as a discrete-time single integrator:
\begin{equation}\label{e1}
x_i(t+1)=x_i(t)+u_i(t),\quad \forall i = {1,\ldots,N}.
\end{equation}
Here, $x_i(t), u_i(t) \in \mathbb{R}^n$ correspond to the state (position) and input vectors. The initial state $x_i(0)$ is deterministic for every agent.

The communication topology is described by an undirected graph $\mathcal{G}(\mathcal{V}, \mathcal{E}, \mathcal{A})$, where $\mathcal{V} = \{1, \ldots, N\}$ and $\mathcal{E} \subseteq \mathcal{V} \times \mathcal{V}$ denote the node and edge sets, respectively. The total number of edges is represented by $N_{\mathcal{E}} = |\mathcal{E}|$, and the edges are labelled as $l_1, \ldots, l_{N_{\mathcal{E}}}$ following the node index order. For any $(j,i) \in \mathcal{E}$, it implies a communication channel exists between node $j$ and node $i$, designating node $j$ as a neighbour of node $i$. The set of neighbours for node $i$ is denoted by $\mathcal{N}_i = \{j \in \mathcal{V} \mid (j, i) \in \mathcal{E}\}$. Regarding the algebraic representation, the adjacency matrix $\mathcal{A} = [a_{ij}] \in \mathbb{R}^{N \times N}$ is defined as $a_{ij} = 1$ if $j \in \mathcal{N}_i$, and $a_{ij} = 0$ otherwise (implying $a_{ii}=0$). The degree matrix is defined as $D_g= \mathrm{diag}\{d_1, \dots, d_N\}$, where $d_i = \sum_{j \in \mathcal{N}_i} a_{ij}$. Consequently, the Laplacian matrix $\mathcal{L} = [l_{ij}] \in \mathbb{R}^{N \times N}$ associated with graph $\mathcal{G}$ is defined as $\mathcal{L} = D_g - \mathcal{A}$. The edge weight matrix is then given by $W \triangleq \mathrm{diag}\{a_{ij}\} = \mathbf{I}_{N_{\mathcal{E}}}$ for the corresponding edges $(i,j) = l_1, \ldots, l_{N_{\mathcal{E}}}$. Furthermore, an arbitrary direction is assigned to each edge on $\mathcal{G}$. Based on this orientation, the incidence matrix $B = [b_{ij}] \in \mathbb{R}^{N_{\mathcal{E}} \times N}$ is defined as follows: $b_{ij} = 1$ if edge $i$ enters node $j$, $b_{ij} = -1$ if it leaves node $j$, and $b_{ij} = 0$ otherwise.

The following assumption and lemma regarding graph connectivity are imposed.
\begin{assumption}\label{ass1}
The topology $\mathcal{G}(\mathcal{V} ,\mathcal{E},\mathcal{A} )$ is a tree.
\end{assumption}
\begin{lemma}[\cite{dimarogonas2010stability}]\label{le1}
Under the condition that $\mathcal{G}$ is a tree, the incidence matrix $B$ is of full row rank.
\end{lemma}

Assumption \ref{ass1} implies a connected, acyclic communication network. This naturally occurs in practical mobile agent scenarios like vehicle platooning or strict leader-follower formations. In the subsequent analysis, Assumption \ref{ass1} is applied directly via the Lemma \ref{le1}'s conclusion, that the incidence matrix $B$ has full row rank.

\subsection{Communication Noise Model}
In practical MAS applications such as vehicular platoons and UAV swarms, inter-agent communication is realised over wireless fading channels, thereby inevitably compromising the information exchange with communication noise \cite{khuwaja2018survey}. To capture this, the communication channel between agents is modelled as a linear Single-Input Single-Output (SISO) link for block-transmitted vector signals. The signal received by agent $i$ from agent $j$, prior to internal processing, is
\begin{equation}\label{e3}
r_{ij}(t) = h_{ij}(t) \left(x_j(t)+\eta_{{ij}_s}(t)\right) + \eta_{{ij}_z}(t),
\end{equation}
where $x_j(t) \in \mathbb{R}^n$ denotes the transmitted state, $h_{ij}(t) \in \mathbb{R}$ represents the scalar channel gain, and $\eta_{{ij}_s}(t)$, $\eta_{{ij}_z}(t)$ represent the impairments occurring within the transmitter and receiver units respectively, which are analytically treated as independent additive white Gaussian noise (AWGN) processes \footnote{The transmitter distortion $\eta_{{ij}_s}$, which encompasses hardware impairments such as quantization errors, phase noise, and DAC nonlinearities, exhibits a significantly larger variance than the receiver thermal noise $\eta_{{ij}_z}$ \cite{bjornson2014massive}.}. The following assumption is imposed on the communication channel.
\begin{assumption}\label{ass2}
    The communication links are modelled as SISO AWGN channels dominated by large-scale fading, where the receiver is assumed to possess perfect channel state information (CSI). Furthermore, vector signals are transmitted in a block-wise manner using orthogonal resources.
\end{assumption}

Under Assumption \ref{ass2}, perfect CSI allows phase compensation at the receiver, $h_{ij}(t)$ is treated as real, non-negative, and random. Orthogonal resources prevent inter-agent interference. Assuming perfect CSI is a standard convention adopted to facilitate the analysis of communication noise \cite{sery2020analog,cao2020optimized}. The randomness in $h_{ij}(t)$ arises from shadow fading caused by environmental obstacles. The large-scale fading model {is classical and well-established in the analysis of} robotic, UAV, and vehicular networks \cite{molisch2009survey,khuwaja2018survey,cai2017low
}.

In accordance with the properties of large-scale fading, the channel gain is derived from an empirical path loss model \cite{khuwaja2018survey} following the linear-domain conversion in \cite[Chapter 2]{goldsmith2005wireless}:
\begin{equation}\label{e5}
    h_{ij}(t)=C_0 \cdot \frac{\mathcal{X}_{ij}(t)}{\|x_i(t) - x_j(t)\|^{\alpha/2}}
\end{equation}
where $C_0 = \sqrt{10^{-P_\mathrm{L}(d_0)/10}} \cdot d_0^{\alpha/2}$ is a constant determined by the reference distance $d_0$ and reference path loss $P_\mathrm{L}(d_0)$. The empirical path loss exponent $\alpha>1$ (with instances of $\alpha \leq 1$ being exceedingly rare) characterises the propagation environment. As reported in {\cite{molisch2009survey, seidel1992914} (see Table 1 in each)}, $\alpha$ typically ranges from 1.6 to 6. For instance, $\alpha$ is 1.61 in urban vehicle-to-vehicle (V2V) channels, $\alpha=1.8$ in V2V highway environments and $\alpha = 2$ corresponds to free space. The log-normal shadowing factor satisfies $\ln(\mathcal{X}_{ij}(t)) \sim \mathcal{N}(0, {(\frac{\ln 10}{20}\sigma_{dB})}^2)$, with $\sigma_{dB}=0$ when $\alpha=2$ \footnote{Strictly, the absence of shadow fading ($\sigma_{dB}=0$, $\mathcal{X}_{ij}=1$) occurs exclusively when $\alpha=2$, as this indicates a free-space propagation environment devoid of obstacle occlusion.}. Both $\alpha$ and $\sigma_{\mathrm{dB}}$ are assumed homogeneous across links and time-invariant. Since packet loss is not considered, $h_{ij}(t)>0$ for all $t\ge0$, $(j,i)\in\mathcal{E}$. For the case where $x_i(t) - x_j(t) \to 0$, $h_{ij}(t) \to \infty$. Based on \eqref{e3}, the impact of $\eta_{{ij}_z}(t)$ becomes negligible, and the perturbation stems solely from $\eta_{{ij}_s}(t)$. This limiting behaviour should be interpreted in the context of the operational range of mobile multi-agent systems, where agents maintain nonzero safety distances. In practice, near-field effects, receiver saturation, and hardware constraints impose a finite upper bound on the channel gain. 

By performing coherent equalization based on channel estimation, the compensated received signal derived from \eqref{e3} is
\begin{equation}\label{e6}
    \hat{x}_{ij}(t) =h_{ij}^{-1}(t)r_{ij}(t)= x_j(t) +\eta_{{ij}_s}(t)+ h_{ij}^{-1}(t) \eta_{{ij}_z}(t),
\end{equation}
where $\eta_{{ij}_s}\sim\mathcal{N}(0,s_{j}\mathbf{I})$ (with $s_{j} > 0$) and $\eta_{{ij}_z}\sim\mathcal{N}(0,z_i\mathbf{I})$ (with $z_i > 0$)  are spatially and temporally independent and identically distributed (i.i.d.) and mutually independent. Defining the aggregate noise $\eta_{ij}(t)\triangleq \eta_{{ij}_s}(t)+h_{ij}^{-1}(t)\eta_{{ij}_z}(t)$, its conditional distribution given the shadowing realisation $X_{ij}(t)$ is
\begin{equation}\label{e8}
    \eta_{ij}(t)\mid\mathcal{X}_{ij}(t) \sim \mathcal{N}\left(0,\left(\frac{z_i \|x_i(t) - x_j(t)\|^{\alpha}}{C_0^2\mathcal{X}_{ij}^2(t)}+s_j\right)\mathbf{I}\right).
\end{equation}

\subsection{Finite-Horizon LQR-Based Cooperative Control}\label{sc}
To coordinate the MAS \eqref{e1} into a target spatial configuration, a distributed finite-horizon LQR framework is considered. Specifically, at each time instant $t$, the control action for agent $i$ is derived by solving a local optimal control problem over a horizon of length $T$:
\begin{subequations}\label{e14}
\begin{align}
\min\limits_{u_i(\cdot)}\quad&  J^\ast_i (x_i(t),u_i(t))\label{eq:14a}\\
\textbf{s.t. } &  \eqref{e1}
\end{align}
\end{subequations}
where the local cost function is defined as
\begin{align}\label{e9}
    J^\ast_i(x_i(t),u_i(t)) &\triangleq \sum\limits_{q=0}^{T-1} \left( \ell_i^{u}(t+q) + \ell_i^{x}(t+q) \right).
\end{align}
with $\ell_i^{u}(t+q) \triangleq \|u_i(t+q)\|_{R_i}^2$ and $\ell_i^{x}(t+q) \triangleq \sum_{j\in \mathcal{N}_i}a_{ij}\|x_i(t+q)-{x}_{j}(t)-d_{ij}\|_{Q_i}^2$. At each time step $t$, agents exchange their current states across the network. The value $x_j(t), j\in\mathcal N_i$, denotes the most recently received state of neighbour $j$. In the local finite-horizon problem solved by agent $i$, this neighbour state is held fixed over the prediction horizon, while only the {$i$-th agent's} own future trajectory $x_i(t+q)$ is optimised. Under the receding-horizon implementation, only the first element of the optimal sequence is applied, after which agents exchange updated states and the local optimisation is solved again at the next time step. The vector $d_{ij}\in \mathbb{R}^n$ represents the desired relative state, satisfying $d_{ij}=-d_{ji}$ for all $(j,i)\in\mathcal{E}$. Furthermore, the diagonal matrices $Q_i \in \mathcal{Q}_i \subset {\mathbb{S}_{++}^{n\times n}}$ and $R_i\in \mathcal{R}_i \subset {\mathbb{S}_{++}^{n\times n}}$ are weights designed to balance the trade-off between cooperative accuracy and control effort, with $\mathcal{Q}_i$ and $\mathcal{R}_i$ denoting their respective finite admissible sets.


From a privacy perspective, the threat model considers an external eavesdropper capable of intercepting transmissions at any location along the communication channels. Furthermore, it is assumed that this adversary possesses complete prior knowledge of the system dynamics, the communication topology, the underlying algorithmic structure, and the formation offsets $d_{ij}$. The eavesdropper's objective is to infer the agents' local control preferences from the observed trajectories, which may allow it to predict agent behaviour, identify vulnerable agents, or induce system miscoordination. According to the noise model in \eqref{e8}, the effective noise in the signal intercepted by the eavesdropper depends on its spatial location. Specifically, the noise variance reaches its minimum at the location of transmitter and its maximum at receiver.

In the LQR-based cooperative control problem considered here, the quadratic form of the cost function in \eqref{e9} is standard and is therefore assumed to be public. Consequently, this work does not regard the algebraic cost structure itself as the private information. Instead, the protected information is the collection of weighting ratios $R_i^{-1}Q_i,\ \forall i\in\mathcal{V}$, which encodes each agent’s relative preference between formation-tracking accuracy and control effort. Protecting $R_i^{-1}Q_i$, rather than $Q_i$ and $R_i$ separately, is also consistent with the scale invariance of the unconstrained LQR formulation \footnote{Replacing $(Q_i,R_i)$ with $(\kappa Q_i,\kappa R_i)$, where $\kappa>0$, scales the local objective by $\kappa$ but leaves the optimiser unchanged}. Thus, uniformly scaled weighting matrices do not represent distinguishable control preferences under the considered mechanism. For analyses of recovery mechanisms for optimal-control weighting parameters, see \cite{jin2021inverse}. The formal adjacency relation and DP guarantee for this weighting-ratio database are introduced next.

\subsection{On Differential Privacy}

   This subsection introduces the definition of DP concerning the weighting ratios, alongside the requisite supporting lemmas. Let $D=\{R_i^{-1}Q_i\}_{i=1}^{N}$ denote the private database comprising the weighting ratios. The universe of all admissible databases is defined as $\mathcal{D}=\{D \mid \forall Q_i\in \mathcal{Q}_i, \forall R_i\in \mathcal{R}_i, i=1,2,\ldots,N\}$.

\begin{definition}[Bounded Adjacency, \cite{dwork2014algorithmic}]\label{de1}
    Two databases $D=\{R_i^{-1}Q_i\}_{i=1}^{N}$ and $D'=\{R_i'^{-1}Q_i'\}_{i=1}^{N}$ are said to be adjacent if there exists a $k_1\in\mathcal{V}$
    \begin{equation}\label{e10}
        \left\{
        \begin{aligned}
            &R_i^{-1}Q_i=R_i'^{-1}Q'_i,\, i\neq k_1, \\
            &\|R_i^{-1}Q_i-R_i'^{-1}Q_i'\| \leq \theta_1,\, i=k_1,
        \end{aligned}
        \right.
    \end{equation}
\end{definition}

A randomised mechanism $\mathcal{M}$ acting on a sensitive database $D$ is said to be differentially private if it ensures that adjacent databases $D$ and $D'$ produce nearly indistinguishable outputs $\mathcal{M}(D)$ and $\mathcal{M}(D')$.

\begin{definition}[\cite{dwork2014algorithmic}]\label{de3}
Let $\epsilon>0$ and $\delta>0$. A randomised algorithm $\mathcal{M}$ with $\mathrm{Range}(\mathcal{M}) \subseteq \mathbb{R}^n$ is $(\epsilon,\delta)$-differentially private if for all $S \subseteq \mathrm{Range}(\mathcal{M})$ and for all $D,D' \in \mathcal{D}$  defined in Definition \ref{de1}:
\begin{equation}\label{e12}
    \Pr[\mathcal{M}(D) \in S] \le \mathbf{e}^\epsilon \Pr[\mathcal{M}(D') \in S] + \delta.
\end{equation}
\end{definition}

\begin{lemma}[\cite{le2014differentially}]\label{le2}
 Let $\epsilon >0$, $\delta\in (0,\frac{1}{2})$ and $\mathcal{Q}(x)=\frac{1}{\sqrt{2\pi}}\int_x^\infty\mathrm{e}^{-\frac{u^2}{2}}du$. Consider a mechanism $M$ with $\ell_2$-sensitivity,
    \begin{align}\label{e13}
        \Delta_{M}=\max\limits_{D,D'} \|M(D)-M(D')\|,
    \end{align}
    where databases $D$ and $D'$ satisfy the adjacent relation as defined in Definition \ref{de1}. Then the mechanism $\mathcal{M}(D)=M(D)+\eta$ is $(\epsilon,\delta)$-differentially private if each component of $\eta$ are i.i.d. and scaled to $\mathcal{N}(0,\sigma^2)$, where $\sigma \ge \frac{\Delta_{M}}{\sqrt{{\mathcal{Q}^{-2}(\delta)}+2\epsilon}-\mathcal{Q}^{-1}(\delta)}$. 
\end{lemma}

\subsection{Problem Statement}
Given the stochastic nature of the communication noise derived previously, the desired spatial configuration control problem is reformulated and formally stated as follows.

\begin{problem}\label{pr1}
Design a \textit{differentially private distributed finite-horizon LQR mechanism} for the MAS described by~\eqref{e1} such that:
\begin{enumerate}
    \item The tracking error $\xi_{l_k}(t)\triangleq x_i(t)-x_j(t)-d_{l_k}$ converges to a finite random limit a.s. and in mean square, $\forall l_k=(i,j)\in\mathcal{E}$;
    \item The mechanism guarantees bounded $(\epsilon,\delta)$-differential privacy from time $t=0$ to infinity for all weighting ratios $R_i^{-1}Q_i$, under the adjacency relation defined in Definition~\ref{de1}.
\end{enumerate}
\end{problem}

The first objective ensures the convergence of the tracking error in a stochastic sense, while the second enforces privacy protection on the weighting ratios. 

\section{Main results}

In this section, a differentially private distributed finite-horizon LQR framework is presented. Within this framework, the convergence of the MAS \eqref{e1} is established. {Specifically, the analysis focuses on the cases $\alpha \leq 2$, which are motivated as they characterize propagation conditions commonly observed in mobile-agent networks. The case $\alpha > 2$, corresponding to more heavily obstructed propagation, are reserved for future research.} Based on the convergence results, the DP guarantees provided by the communication noise \eqref{e8} for the weighting ratios $R_i^{-1}Q_i$, $\forall i \in \mathcal{V}$, are then rigorously derived. Furthermore, a sensitivity-based justification is provided to validate the rationale for protecting weighting ratios rather than gradients.

\subsection{Differentially Private LQR Cooperative Mechanism}

Building upon the LQR-based cooperative framework in Section \ref{sc}, the differentially private mechanism is established by introducing communication noise \eqref{e8} and modifying cost function \eqref{eq:14a} as follows: 
\begin{align}\label{e15}
    J_i(x_i(t),u_i(t),c(t)) &\triangleq \sum\limits_{q=0}^{T-1} \left( {\ell}_i^{u}(t\!+\!q) \!+ \!c(t)\hat{\ell}_i^{x}(t\!+\!q) \right),
\end{align}
with $\hat{\ell}_i^{x}(t+q)\triangleq \sum\limits_{j\in \mathcal{N}_i}a_{ij}\|x_i(t+q)-\hat{x}_{ij}(t)-d_{ij}\|_{Q_i}^2,$ $c(t)>0$ a time-varying penalty factor and $c(t)\rightarrow0$ as $t\rightarrow\infty$, $\hat{x}_{ij}(t)$ the signal received by agent $i$ according to the communication model \eqref{e6}. 
The sequence $c(t)$ is treated as a globally design parameter and provides an adjustable trade-off between tracking performance and privacy. In principle, although the physical communication-noise variance in \eqref{e8} is determined by the channel properties and is therefore not designable, the penalty factor $c(t)$ scales the cooperative-error term in \eqref{e15} and the control input in \eqref{e16}. As a result, $c(t)$ changes the sensitivity of the mechanism to the protected weighting ratios and therefore indirectly shapes the privacy–performance trade-off.

Following the receding horizon control principle, only the initial element of the computed optimal sequence in \eqref{e14} and \eqref{e15} is applied to the system \eqref{e1} at each step. From an analytical perspective, the unconstrained control input $u_i(t)$ for the optimisation problem in \eqref{e14} is then derived following the explicit MPC \cite{bemporad2002model}. Combining \eqref{e1}, \eqref{e14} and \eqref{e15}, the control sequence $U_i^{\ast}=\left[\begin{matrix}
    u_i^{\ast}(t)^\top & u_i^{\ast}(t+1)^\top &\ldots &u_i^{\ast}(t+T-1)^\top
\end{matrix}\right]^\top$ is obtained as
   \begin{align*}
U_i^\ast&=c(t)\sum\limits_{j\in\mathcal{N}_i}a_{ij}\left(\bar{R}_i+c(t)\sum\limits_{j\in\mathcal{N}_i}a_{ij}\Phi^\top\bar{Q}_i\Phi\right)^{-1}\\
&\quad\times\Phi^\top\bar{Q}_i\Lambda(\hat{x}_{ij}(t)-x_i(t)+d_{ij}),
\end{align*}
where $\bar{Q}_i=\mathrm{diag}\{\underbrace{Q_i,\ldots,Q_i}_T\}$, $\bar{R}_i=\mathrm{diag}\{\underbrace{R_i,\ldots,R_i}_T\}$, $\Phi \in \mathbb{R}^{nT\times nT}$, $\Lambda \in \mathbb{R}^{nT\times n}$ and
\begin{align*}
    \Phi=\left[\begin{matrix}
        \mathbf{I}_{n\times n} & \mathbf{0}_{n\times n} &\ldots & \mathbf{0}_{n\times n}\\
       \mathbf{I}_{n\times n} & \mathbf{I}_{n\times n} &\ldots & \mathbf{0}_{n\times n} \\
        \vdots & \vdots &\ddots & \vdots  \\
        \mathbf{I}_{n\times n} & \mathbf{I}_{n\times n}& \ldots & \mathbf{I}_{n\times n} 
    \end{matrix}\right],\ \Lambda=\left[\begin{matrix}\mathbf{I}_{n\times n}\\\mathbf{I}_{n\times n}\\ \vdots \\\mathbf{I}_{n\times n} \end{matrix}\right].
\end{align*} Define 
\begin{align}\label{e16}
\left[
    \begin{matrix}
     \psi_{1,t} \\ \psi_{2,t} \\ \vdots \\\psi_{T,t}  
\end{matrix}\right]\triangleq\left(\bar{R}_i+c(t)\sum\limits_{j\in\mathcal{N}_i}a_{ij}\Phi^\top\bar{Q}_i\Phi\right)^{-1}\Phi^\top\bar{Q}_i\Lambda,
\end{align} $\psi_{1,t} \in \mathbb{R}^{n\times n}$. The input $u_i(t)$ is given as
\begin{align}\label{e17}
    u_i(t) &= u_i^{\ast}(t)=c(t)K_{i,t}\sum\limits_{j\in\mathcal{N}_i}a_{ij}  (\hat{x}_{ij}(t)+d_{ij}-x_i(t)),\nonumber\\
    K_{i,t}&=\psi_{1,t}.
\end{align}

The complete distributed execution protocol is outlined in Algorithm \ref{alg1}. In this framework, Algorithm~\ref{alg1} can be viewed as a sequential composition of control mechanisms, denoted by $\mathbf{\mathcal{M}}=\{\mathbf{\mathcal{M}}_0, \mathbf{\mathcal{M}}_1, \ldots\}$. The joint mechanism $\mathbf{\mathcal{M}}_t=\{\mathbf{\mathcal{M}}_{t,1}, \mathbf{\mathcal{M}}_{t,2}, \ldots, \mathbf{\mathcal{M}}_{t,N}\}$ characterises the execution of Steps~4–7 at time~$t$. Specifically, each local mechanism $\mathbf{\mathcal{M}}_{t,i}$ yields the updated states $x_i(t+1)$ for all $i\in \mathcal{V}$. Owing to the recursive nature of the algorithm, $\mathbf{\mathcal{M}}_t$ exhibits a dependence on the historical sequence $\{\mathbf{\mathcal{M}}_0, \mathbf{\mathcal{M}}_1, \ldots, \mathbf{\mathcal{M}}_{t-1}\}$. The selection rule and analytical properties of $c(t)$ will be discussed later.

\begin{algorithm}
\caption{Differentially Private Cooperative Control Algorithm with Inherent Communication Noise}
\begin{algorithmic}[1]\label{alg1}
\STATE \textbf{Initialisation:} $x_i(0)\in \mathbb{R}^n$, $Q_i \in \mathcal{Q}_i$, $R_i\in \mathcal{R}_i$, $i=1,2,\ldots,N$, and $d_{ij}$, $(i,j)=l_1,l_2,\ldots,l_{N_\mathcal{E}}$.
\STATE \textbf{Choose} the public penalty sequence $\{c(t)\}_{t\ge 0}$ for all agents $i=1,2,\ldots,N$.
\FOR{$t=0, 1, 2, \ldots$}
    \FOR{$i=1, 2, \ldots, N$}
    \STATE Agent $i$ receives $\hat{x}_{ij}(t)$ in \eqref{e6} from its neighbour.
    \ENDFOR 
    \STATE Each agent $i$ derives the control input $u_i(t)$ by computing \eqref{e17}, and updates its state $x_i(t+1)$ according to \eqref{e1}.
    \ENDFOR
\end{algorithmic}
\label{alg:dpcc}
\end{algorithm}

\begin{remark}\label{re2}
    In \eqref{e15}, the cooperative error penalty factor $c(t)$ is designated as a globally shared parameter among all agents. $c(t)$ could theoretically be designed as a private parameter $c_i(t)$. However, doing so would expand the scope of the sensitive dataset. Furthermore, given the time-varying nature of $c(t)$, continuously providing privacy guarantees for it across iterations would impose an excessive and redundant privacy burden. Thus, treating $c(t)$ as public information avoids unnecessary noise accumulation. 
\end{remark}

\begin{remark}
The frozen-neighbour formulation is adopted deliberately in the proposed finite-horizon LQR framework (see \eqref{e9} and \eqref{e15}). Agent $i$ has no access to the future inputs or predicted trajectories of its neighbours, and the alternative of exchanging assumed trajectories, as in distributed MPC schemes (e.g., \cite{bemporad2002model}). However, this would require each agent to transmit $T$ predicted states at every time step. Since these predictions depend on the private weighting ratio through $K_{i,t}$, such an exchange would inflate the cumulative privacy budget via composition, in addition to increasing the communication burden.
\end{remark}

\subsection{Convergence Analysis}

Note that $\alpha = 2$ corresponds to an obstacle-free communication environment, which implies that the shadowing effect term $\mathcal{X}_{ij}(t)$ in \eqref{e8} is absent. Consequently, the convergence analysis is conducted separately for the cases of $\alpha = 2$ and $\alpha < 2$. 

First, let the edge-related errors $\xi_{l_k}$ in Problem \ref{pr1} be stacked as $\xi=[\xi^\top_{l_1},\xi^\top_{l_2},\ldots,\xi^\top_{l_{N_\mathcal{E}}}]^\top$. In contrast to the formulation $\eta_{ij}(t) = -\eta_{ji}(t)$ in \cite{ma2026inherent}, the noises $\eta_{ij}(t)$ and $\eta_{ji}(t)$ in this work are mutually independent and drawn from distinct distributions according to \eqref{e8}. Define $\eta_i(t) = \sum\limits_{j \in \mathcal{N}_i} \eta_{ij}(t)$ and stack the aggregated noise term as $\eta=[\eta^\top_{1},\eta^\top_{2},\ldots, \eta^\top_{{N}}]^\top$. Based on \eqref{e1} and \eqref{e17}, the dynamics of $\xi(t)$ are governed by:
\begin{align}\label{e23}
\xi(t+1) &= \xi(t) - c(t)(B\otimes \mathbf{I}_n)K_t(B^\top W\otimes\mathbf{I}_n)\xi(t)\nonumber\\
&\quad +c(t)(B\otimes \mathbf{I}_n)K_t\eta(t).
\end{align} Here, $B$ and $W$ denote the incidence matrix and the edge weight matrix, respectively, while $K_t = \mathrm{diag}\{K_{1,t},\ldots,K_{N,t}\}$. Given that $Q_i$ and $R_i$ are drawn from the finite matrix sets $\mathcal{Q}_i$ and $\mathcal{R}_i$, and that $c(t) \to 0$ as $t \to \infty$, it can be deduced from \eqref{e16} and \eqref{e17} that $\|K_{i,t}\|$ is uniformly bounded for all $i\in\mathcal{V}$ and $t\ge0$. Moreover, since $Q_i$ and $R_i$ are diagonal matrices for all $i\in\mathcal{V}$, their positive definiteness and commutativity guarantee that $K_{i,t}$ is a positive definite matrix, and thus $K_t$ is also positive definite. Define $\Psi_t \triangleq (B\otimes \mathbf{I}_n)K_t(B^\top W\otimes\mathbf{I}_n)$. The norm of $\Psi_t$ is then bounded by a positive constant $\rho_\Psi$ for all $t\ge0$, i.e., $\|\Psi_t\| \leq \rho_\Psi$. Furthermore, considering that $K_t$ is positive definite and $W=\mathbf{I}_{N_\mathcal{E}}$, combined with the full row rank property of $B$ established in Lemma \ref{le1}, it follows that $\Psi_t$ is positive definite. By virtue of the boundedness of $K_{i,t}$, let $\lambda_{\min}^{\Psi} \triangleq \inf\limits_{t > 0} \lambda_{\min}(\Psi_t)<\infty$ denotes the infimum of the minimum eigenvalues of $\Psi_t$ over the time horizon $t\ge0$. The strict positivity $\lambda_{\min}^{\Psi} > 0$ is guaranteed by the positive definiteness of $\Psi_t$.

{The convergence properties of $\xi(t)$ are first established for the case $\alpha = 2$, as part of the resulting analysis will subsequently be used to facilitate the convergence analysis for $\alpha < 2$.} Prior to establishing the main convergence results, the following lemmas are introduced.
\begin{lemma}\cite[Lemma A.4]{li2018distributed}\label{le3}
    Let $\{X(k),\mathcal{F}(k)\}$ be a martingale sequence satisfying $\sup\limits_{k\geq0}\mathbb{E}\|X(k)\|^2< \infty$. Then $X(k)$ converges a.s. and in mean square.
\end{lemma}
\begin{lemma}\cite{polyak1987introduction}\label{le4} 
Let $\{x(k)\}$, $\{\beta(k)\}$, and $\{\alpha(k)\}$ be real sequences defined for $k \ge 0$ that satisfy $x(k+1) \le (1 - \beta(k))x(k) + \alpha(k)$. Provided that these sequences and coefficients satisfy $0\le x(k) $, $0 < \beta(k) \le 1$ for $k_1\le k\ (k_1<\infty)$, $\sum\limits_{k=0}^{\infty} \beta(k) = \infty$, and $\lim\limits_{k \to \infty} \frac{\alpha(k)}{\beta(k)} = 0$, then $\lim\limits_{k \to \infty} x(k) = 0$.
\end{lemma}

\begin{theorem}\label{th2}
  Suppose that Assumptions~\ref{ass1} and~\ref{ass2} hold for the MAS \eqref{e1} with $\alpha=2$ in \eqref{e8}. If the $c(t)$ in \eqref{e15} is selected to satisfy $c(t) \in \ell_1$ and $c(t)>0,\, \forall t\ge0$, then Algorithm~\ref{alg1} ensures that the tracking error $\xi(t)$ converges both a.s. and in mean square.  
\end{theorem}

\begin{proof}
    Following from \eqref{e23},
    \begin{align}\label{e24}
    \xi(t+1)&=(\mathbf{I}_{n\cdot N\mathcal{E}}-c(t)\Psi_t)\xi(t)+c(t)(B\otimes \mathbf{I}_n)K_t\eta(t)\nonumber\\
&=\Phi(t,0)\xi(0)+\sum\limits_{q=0}^t\Phi(t,q+1)c(q)\mathcal{B}K_q\eta(q),\end{align}where the state transition matrix is defined as $\Phi(t,q)\triangleq\prod\limits_{j=q}^t(\mathbf{I}_{n\cdot N_\mathcal{E}}-c(j)\Psi_j)$ for $t\ge q$, with $\Phi(t,t+1)=\mathbf{I}$ and $\mathcal{B}=B\otimes \mathbf{I}_n$. The first term of \eqref{e24} is deterministic and depends on the initial state, whereas
the second term $$M(t)=\sum\limits_{q=0}^t\Phi(t,q+1)c(q)\mathcal{B}K_q\eta(q),$$
is stochastic. To facilitate the convergence analysis, the auxiliary process $\bar{M}(t)\triangleq \Phi(\infty, t+1) M(t)=\sum\limits_{q=0}^t \Phi(\infty, q+1) c(q)\mathcal{B}K_q \eta(q)$ is introduced since $\Phi(\infty, q+1) = \Phi(\infty, t) \Phi(t, q+1)$. It is now established that the sequence $\{\bar{M}(t)\}_{t\ge0}$ constitutes a martingale with respect to the natural filtration $\mathcal{F}_t=\sigma(\eta(0),\ldots,\eta(t))$. Since $c(q)$, $\Psi_q$, and $\Phi(\infty,q+1)$ are deterministic and uniformly bounded, the sequence is integrable at any given step (i.e., $\mathbb{E}[\|\bar{M}(t)\|] < \infty$). Second, note that $\eta(t+1)$ is zero-mean and independent of the history $\mathcal{F}_t$, there holds $\mathbb{E}[\bar{M}(t+1)\mid\mathcal{F}_t]=\bar{M}(t)$. With the martingale property verified, the uniform boundedness of this sequence is subsequently analysed as follows:

\begin{align}\label{e25}
    &\sup\limits_{t\geq0}\mathbb{E}[\|\bar{M}(t)\|^2]=\sup\limits_{t\geq0}\mathbb{E}[\bar{M}^\top(t)\bar{M}(t)] \nonumber\\
    &=\sup\limits_{t\geq0}\sum\limits_{q=0}^t\!\mathbb{E}[c^2(q)\eta^\top\!(q)K^\top_q\mathcal{B}^\top\Phi^\top\!(\infty,q\!+\!1)\Phi(\infty,q\!+\!1)\mathcal{B}K_q\eta(q)] \nonumber\\
    &\leq \sup\limits_{t\geq0} \sum\limits_{q=0}^t\rho^2_{K\mathcal{B}} c^2(q)\|\Phi(\infty,q+1)\|^2\mathbb{E}[\eta^\top(q)\eta(q)],
\end{align} where $\|K_q\|\|\mathcal{B}\|\le\rho_{K\mathcal{B}}$, $\forall q\ge0$. The second equality holds because for any $q_1 \neq q_2$, $\eta(q_1)$ and $\eta(q_2)$ are mutually independent and satisfy $\mathbb{E}[\eta(q_1)]=\mathbb{E}[\eta(q_2)]=0$. Since $\Psi_t$ is positive definite and $c(t)\rightarrow0$ as $t \to \infty$, a time instant $t_1 > 0$ can be identified such that for all $t \geq t_1$, the matrix $\mathbf{I}-c(t)\Psi_t$ remains positive definite and $\|\mathbf{I}-c(t)\Psi_t\| < 1$. It subsequently follows that $\|\Phi(\infty,t_1)\|^2 \leq \prod_{j=t_1}^\infty \|\mathbf{I}_{nN_\mathcal{E}}-c(j)\Psi_j\|^2 < 1$. This inequality guarantees the existence of a uniform bound $\rho_\Phi > 0$ such that $\|\Phi(t,q)\| \leq \rho_\Phi$ holds for all $t \ge q\ge0$. Based on \eqref{e8}, with the shadowing effect $\mathcal{X}_{ij}(t)$ omitted and $\alpha=2$,
\begin{align}\label{e26}
        \mathbb{E}[\eta^\top(t)\eta(t)]&=\mathbb{E}[\sum\limits_{i={1}}^{N}\eta_i^\top(t)\eta_i(t)]\nonumber\\
        &=\mathbb{E}[\sum\limits_{i={1}}^{N}\sum\limits_{j\in\mathcal{N}_i}\eta_{ij}^\top(t)\eta_{ij}(t)]\nonumber\\
        &=n \sum\limits_{i={1}}^{N}\sum\limits_{j\in\mathcal{N}_i}\mathbb{E}[\frac{z_{i}\|x_i(t)-x_j(t)\|^2}{C_0^2}+s_{j}]\nonumber\\
        &\leq n\sum\limits_{i={1}}^{N}\sum\limits_{j\in\mathcal{N}_i}\mathbb{E}[\frac{2z_{i}(\|\xi_{ij}(t)\|^2+\|d_{ij}\|^2)}{C_0^2}+s_{j}]\nonumber\\
        &= n\!\sum\limits_{(i,j)\in\mathcal{E}}\!\mathbb{E}[\frac{2(z_i\!+\!z_j)(\|\xi_{ij}(t)\|^2\!+\!\|d_{ij}\|^2)}{C_0^2}\!+\!s_{i}\!+\!s_{j}]\nonumber\\
        &\leq  \frac{2nzd^\top d+2nz\mathbb{E}[\xi^\top(t)\xi(t)]}{C_0^2}+nsN_\mathcal{E}
    \end{align}
 where the second equality holds since $\eta_{ij}$ are mutually independent across all links $(i,j)$, the fourth inequality follows from the triangle inequality and $z = \max \{z_i + z_j \mid (i,j) \in \mathcal{E}\}$, $s = \max \{s_i + s_j \mid (i,j) \in \mathcal{E}\}$, and $d=[d^\top_{l_1},\ldots,d^\top_{l_{N_\mathcal{E}}}]^\top$.

 Following from \eqref{e24}, 
 \begin{align}\label{e27}
\xi^\top(t+1)&\xi(t+1)=\xi^\top(t)\xi(t)-2c(t)\xi^\top(t)\Psi_t\xi(t)\nonumber\\
&+c^2(t)\xi^\top(t)\Psi_t^\top\Psi_t\xi(t)+c^2(t)\eta^\top(t)K_t^\top\mathcal{B}^\top\mathcal{B}K_t\eta(t)\nonumber\\
     &+c(t)\xi^\top(t)(\mathbf{I}_{n\cdot N\mathcal{E}}-c(t)\Psi_t^\top)\mathcal{B}K_t\eta(t)\nonumber\\
&+c(t)\eta^\top(t)K_t^\top\mathcal{B}^\top(\mathbf{I}_{n\cdot N\mathcal{E}}-c(t)\Psi_t^\top)\xi(t).
 \end{align}
By taking the expectation operator to both sides and applying the Double Expectation Theorem \cite{billingsley2012probability}, there holds 
\begin{align*}
    &\mathbb{E}\left[\mathbb{E}[c(t)\xi^\top(t)(\mathbf{I}_{n\cdot N\mathcal{E}}-c(t)\Psi_t^\top)\mathcal{B}K_t\eta(t)]\mid\xi(t)\right] = 0,\\
    &\mathbb{E}\left[\mathbb{E}[c(t)\eta^\top(t)K_t^\top\mathcal{B}^\top(\mathbf{I}_{n\cdot N\mathcal{E}}-c(t)\Psi_t^\top)\xi(t)]\mid\xi(t)\right] = 0.
\end{align*} Consequently,
 \begin{align}\label{e28}
     &\mathbb{E}[\xi^\top(t+1)\xi(t+1)]
     =\mathbb{E}[\xi^\top(t)\xi(t)]-2c(t)\mathbb{E}[\xi^\top(t)\Psi_t\xi(t)]\nonumber\\
     &\quad+c^2(t)\mathbb{E}[\xi^\top(t)\Psi_t^\top\Psi_t\xi(t)]+c^2(t)\mathbb{E}[\eta^\top(t)K_t^\top\mathcal{B}^\top\mathcal{B}K_t\eta(t)]\nonumber\\
     &\leq \!(1\!-\!2c(t)\lambda_{\min}^{\Psi}\!+\!c^2(t)\rho_{\Psi}^2)\mathbb{E}[\xi^\top(t)\xi(t)]\!+\!c^2(t)\rho_{K\mathcal{B}}^2\mathbb{E}[\eta^\top(t)\eta(t)].
 \end{align}
 Substituting \eqref{e26} into \eqref{e28},
 \begin{align}\label{e29}
        &\mathbb{E}[\xi^\top(t\!+\!1)\xi(t\!+\!1)]\nonumber\\
        &\leq(1\!-\!2c(t)\lambda_{\min}^{\Psi}+c^2(t)\rho_{\Psi}^2
    +\frac{2nzc^2(t)\rho_{K\mathcal{B}}^2}{C_0^2}) \mathbb{E}[\xi^\top(t)\xi(t)]\nonumber\\
    &\quad+\frac{2nz c^2(t)\rho_{K\mathcal{B}}^2d^\top d}{C_0^2}+nsN_\mathcal{E}c^2(t)\rho_{K\mathcal{B}}^2\nonumber\\ 
    &\leq\! \phi(t,0)\xi^\top\!(0)\xi(0)\!\!+\!\!\sum\limits_{q=0}^t\phi(t,q\!+\!1)c^2(q)\rho_{K\mathcal{B}}^2[\frac{2nz d^\top d}{C_0^2}+nsN_\mathcal{E}],
    \end{align}
where $\phi(t,q)=\prod\limits_ {j=q} ^t(1\!-\!2c(j)\lambda_{\min}^{\Psi}\!+\!c^2(j)\rho_{\Psi}^2
    +\frac{2nzc^2(j)\rho_{K\mathcal{B}}^2}{C_0^2})$, when $t\ge q\ge0$, and $\phi(t,t+1)=1$. Applying the inequality $1+w \le e^{w},\,\forall w\in\mathbb{R}$, for all $t\ge q\ge 0$
\begin{multline}\label{e30}
\phi(t,q)
=\prod_{j=q}^{t}\!\left(1 - 2c(j)\lambda_{\min}^{\Psi} + c^2(j)\rho_{\Psi}^2 + \frac{2nzc^2(j)\rho_{K\mathcal{B}}^2}{C_0^2}\right)
\\
\leq
\exp\!\Bigg(\sum_{j=q}^{t}\!\!\left[-2c(j)\lambda_{\min}^{\Psi} + c^2(j)\rho_{\Psi}^2 + \frac{2nzc^2(j)\rho_{K\mathcal{B}}^2}{C_0^2}\right]\!\Bigg).
\end{multline}
The validity of this inequality also relies on the fact that
\begin{align*}
    1&-2c(j)\lambda_{\min}^{\Psi}+c^2(j)\rho_{\Psi}^2
    +\frac{2nzc^2(j)\rho_{K\mathcal{B}}^2}{C_0^2}\\
    & =(1-c(j)\lambda_{\min}^{\Psi})^2\!+\!c^2(j)(\rho_{\Psi}^2-(\lambda_{\min}^{\Psi})^2)\!+\!\frac{2nzc^2(j)\rho_{K\mathcal{B}}^2}{C_0^2}>0,
\end{align*}where $\rho_{\Psi}^2-(\lambda_{\min}^{\Psi})^2\ge0$ given $\lambda_{\min}^{\Psi}\le\|\Psi(t)\|\le{\rho}_{\Psi}, \forall t\ge0$.
It follows from the conditions $c(t)\in\ell_1$ in Theorem \ref{th2} that $\sum_{j=0}^{\infty} c^2(j) < \infty$ since $\ell_1\subset\ell_2$. Together with $c(t)>0$ and $c(t)\to0$ as $t\to \infty$, this ensures that there exists a constant $\rho_{\phi}>0$ such that
\[
\phi(t,q)\le \rho_{\phi},\quad \forall\, t\ge q\ge0.
\]
Combining \eqref{e29}, \eqref{e30} and $\sum_{j=0}^{\infty} c^2(j)\!<\!\infty$, yields 
\begin{align*}
    \mathbb{E}&[\xi^\top(t\!+\!1)\xi(t\!+\!1)]\leq \rho_\xi<\infty,\, \forall t\ge0,
\end{align*} 
where 
\begin{align}\label{e31}
    \rho_\xi\triangleq\rho_{\phi}\xi^\top\!(0)\xi(0)+\rho_{\phi}\rho_{K\mathcal{B}}^2[\frac{2nz d^\top d}{C_0^2}+nsN_\mathcal{E}]\sum\limits_{q=0}^\infty c^2(q).
\end{align} 
Substituting this relation back into \eqref{e26} and \eqref{e25} leads to the result
\begin{align*}
    \sup\limits_{t\geq0}\mathbb{E}[\|\bar{M}(t)\|^2]< \infty.
\end{align*}
By Lemma \ref{le3}, $\bar{M}(t)$ converges a.s. and in mean square to a finite random limit $$M^*=\sum_{q=0}^{\infty}\Phi(\infty,q+1)c(q)\mathcal{B}K_q\eta(q).$$

It is now demonstrated that the deviation $M(t)-\bar{M}(t)$ converges to $0$ a.s. and in mean square. By substituting the definitions of $M(t)$ and $\bar{M}(t)$, $M(t)-\bar{M}(t)= (\mathbf{I}-\Phi(\infty, t)) M(t)$. Similarly, following the steps from \eqref{e25} to \eqref{e31} yields  
\begin{align*}
    \sup\limits_{t\geq0}\mathbb{E}[\|{M}(t)\|^2]< \infty.
\end{align*}
The limit in mean square is given 
\begin{align*}
     \lim_{t \to \infty}& \mathbb{E}[ \|M(t) - \bar{M}(t)\|^2 ]\\ &\leq \lim_{t \to \infty} \|\mathbf{I}-\Phi(\infty, t)\|^2\sup\limits_{t\geq0}\mathbb{E}[\|{M}(t)\|^2]\\
     &=\lim_{t \to \infty}\left\| \mathbf{I}-\prod_{j=t}^{\infty} (\mathbf{I} - c(j)\Psi_j)\right \|^2\sup\limits_{t\geq0}\mathbb{E}[\|{M}(t)\|^2]\\
     &=\lim_{t \to \infty}\left\|\sum_{j=t}^{\infty} c(j)\Psi_j \prod_{k=t}^{j-1} (\mathbf{I} - c(k)\Psi_k)\right\|^2\sup\limits_{t\geq0}\mathbb{E}[\|{M}(t)\|^2]\\
     &\leq\lim_{t \to \infty}\left(\sum_{j=t}^{\infty} c(j)\right)^2\rho_{\Psi}^2\sup\limits_{t\geq0}\mathbb{E}[\|{M}(t)\|^2],
\end{align*}
where the fourth inequality is obtained by noting that $\|\mathbf{I}-c(t)\Psi_t\| < 1$ for a sufficiently large $t>t_1$. Given that $c(t) \in \ell_1$, it holds that $\lim\limits_{t \to \infty}\left(\sum\limits_{j=t}^{\infty} c(j)\right)^2 = 0$, from which it follows directly that $\lim\limits_{t \to \infty} \mathbb{E}\left[ \|M(t) - \bar{M}(t)\|^2 \right] = 0$. Applying Minkowski's inequality confirms that $M(t)$ converges in mean square to the identical limit $M^*$. On the other hand, given the relation $\bar{M}(t) = \Phi(\infty,t+1) M(t)$, where $\Phi(\infty,t+1)$ is deterministic, the a.s. convergence of $\bar{M}(t)$ to $M(t)$ follows immediately from the condition $c(t) \in \ell_1$, which  guarantees that $\lim\limits_{t\to \infty}\Phi(\infty,t+1) = \lim\limits_{t\to \infty}\prod\limits_{j=t}^\infty(\mathbf{I}-c(j)\Psi_j) = \mathbf{I}$. By invoking the additivity of almost sure limits, ${M}(t)$ converges to ${M}^*$ almost surely. Consequently, it is established that ${M}(t)$ converges to ${M}^*$ a.s. and in mean square.

Synthesising the aforementioned analyses, from \eqref{e24},
\begin{equation}\label{e32}
\xi(t)\to \xi^* := \Phi(\infty,0)\,\xi(0) + M^*, \quad \text{a.s.\ and in mean square,}
\end{equation}
where $\Phi(\infty,0)$ exists 
since $\prod_{j=0}^\infty (I - c(j)\Psi_j)$ converges under $c(j)\in\ell_1$ and $\Psi_j\succ0$. More specifically, the properties of the limiting variable $\xi^*$ are characterised as follows.

The convergence of the first-order moment is addressed first. 
\begin{align*}
    \mathbb{E}[\xi^*]&=\Phi(\infty,0)\xi(0)+\mathbb{E}[M^*]=\Phi(\infty,0)\xi(0).
\end{align*} Moreover, since $c(t) \to 0$ as $t \to \infty$, there exists a finite integer $t_2 \geq 0$ such that $c(t)\lambda_{\min}^{\Psi} \leq 1$ for all $t \geq t_2$. Given that the inequality $1 - x \leq \exp(-x)$ holds for $x \in (0, 1]$, it follows that:
\begin{align*}
    \|\mathbb{E}[\xi^*]\|&\leq \|\Phi(\infty,0)\,\|\cdot\|\xi(0)\|\\
   & \leq \exp\left( -\lambda_{\min}^{\Psi} \sum\limits_{j=t_2}^\infty c(j)\right)\|\xi(0)\|\prod_{j=0}^{t_2-1} \|I - c(j)\Psi_j\|.
\end{align*}

Proceeding to the steady-state covariance
\begin{align*}
   &\mathrm{Cov}[\xi^*]= \mathbb{E}[(\xi^*-\mathbb{E}[\xi^*])(\xi^*-\mathbb{E}[\xi^*])^\top]\\
&=\mathbb{E}\left[\sum_{q=0}^{\infty}c^2(q)\Phi(\infty,q\!+\!1)\mathcal{B}K_q\eta(q)\eta^\top(q)K_q^\top\mathcal{B}^\top\Phi^\top(\infty,q\!+\!1)\right]\\
&=\sum_{q=0}^{\infty}c^2(q)\Phi(\infty,q\!+\!1)\mathcal{B}K_q\mathrm{Cov}[\eta(q)]K_q^\top\mathcal{B}^\top\Phi^\top(\infty,q\!+\!1),
\end{align*}where 
\begin{align*}
    \mathrm{Cov}[\eta(q)]=\mathrm{diag}\{\mathrm{Cov}[\eta_{1}(q)],\, \mathrm{Cov}[\eta_{2}(q)],\, \ldots,\,\mathrm{Cov}[\eta_{N}(q)]\},
\end{align*} as $\mathbb{E}[\eta_m(t)\eta_n^\top(t)]=0, \forall n\neq m, n,m\in\{1,2,\ldots,N\}$. Moreover, $\mathrm{Cov}[\eta_{i}(q)]=\sum\limits_{j\in\mathcal{N}_i}(\frac{z_i\|x_i(t) - x_j(t)\|^{2}}{C_0^2}+s_j)\mathbf{I}, \forall i\in \mathcal{V}$. Observing that $\mathbb{E}[\xi^\top(t)\xi(t)]$ is bounded, it follows by definition that $\mathbb{E}[\|x_i(t) - x_j(t)\|^2]$ is bounded $\forall(i,j)\in\mathcal{E}$, which, together with the condition $c(q) \in \ell_2$ since $\ell_1\subset\ell_2$, confirms that 
\begin{align*}
    \mathrm{Cov}[\xi^*]<\infty.
\end{align*}

Lastly, the mean square deviation is examined.
\begin{align*}
    \mathbb{E}[\|\xi^*\|^2] &= \mathbb{E}[\xi^{*\top}\xi^*]\\
    &=\mathbb{E}[\xi^\top(0)\Phi^\top(\infty,0)\Phi(\infty,0)\xi(0)]+\mathbb{E}[M^{\ast\top}M^\ast]\\
    &\leq \exp\left( -2\lambda_{\min}^{\Psi} \sum\limits_{j=t_2}^\infty c(j)\right)\|\xi(0)\|^2\prod_{j=0}^{t_2-1} \|I - c(j)\Psi_j\|^2\\
    &+\rho^2_{K\mathcal{B}} \rho_\Phi^2\left[\frac{2nzd^\top d+2nz\rho_\xi}{C_0^2}\!+\!nsN_\mathcal{E}\right]\sum\limits_{j=0}^\infty c^2(j)<\infty,
\end{align*}
where $\rho_\xi$ is defined in \eqref{e31}. The proof is completed.
\end{proof}

\begin{corollary} \label{co1}
    Under the conditions of Theorem~\ref{th2}, if the penalty factor is selected such that $c(t)\in\ell_2$, $c(t) \notin \ell_1$ and $c(t)\to 0, t\to\infty$, then $\xi(t)$ converges to $0$ in mean square.
\end{corollary}
\begin{proof}
Given $c(t)\to 0$, there exists a sufficiently large $t_3>0$ such that $\forall t\ge t_3$, $2\lambda_{\min}^{\Psi}-c(t)\rho_{\Psi}^2-\frac{2nzc(t)\rho_{K\mathcal{B}}^2}{C_0^2}>0$ and $c(t)(2\lambda_{\min}^{\Psi}-c(t)\rho_{\Psi}^2-\frac{2nzc(t)\rho_{K\mathcal{B}}^2}{C_0^2})\leq1$. Note that $\sum\limits_{t=0}^\infty2c(t)\lambda_{\min}^{\Psi}-c^2(t)\rho_{\Psi}^2-\frac{2nzc^2(t)\rho_{K\mathcal{B}}^2}{C_0^2}=\infty$ and $\frac{{2nz c(t)\rho_{K\mathcal{B}}^2d^\top d}+nsN_\mathcal{E}c(t)C_0^2\rho_{K\mathcal{B}}^2}{2C_0^2\lambda_{\min}^{\Psi}-c(t)C_0^2\rho_{\Psi}^2-{2nzc(t)\rho_{K\mathcal{B}}^2}}\to0$ as $t\to\infty$ since $c(t)\in \ell_2,\notin\ell_1$ and $c(t)\to0, t\to\infty$. Recall \eqref{e29} and follow Lemma \ref{le4}, $\lim\limits_{t\to \infty} \mathbb{E}[\xi^\top(t\!+\!1)\xi(t\!+\!1)]=0$. The proof is completed.
\end{proof} 
Next, the convergence of the system is analysed for the case where $1<\alpha < 2$. Given that the distribution presented in \eqref{e8} is conditional on $\mathcal{X}_{ij}(t)$, the expectation and variance of $\eta_{ij}(t)$ are analysed first.
\begin{lemma}\label{le5}
    Consider a scalar random variable $y$ and a random vector $x \in \mathbb{R}^n$. Suppose $y$ is log-normally distributed with $\ln(y) \sim \mathcal{N}(0, b^2)$, and the conditional distribution of $x$ given $y$ is
    \begin{align*}
        x \mid y \sim \mathcal{N}\left(0, \left(\frac{a}{y^2} + c\right)\mathbf{I}\right),
    \end{align*}
    where $a, b, c > 0$ are positive constants. Then, the unconditional expectation and covariance of $x$ are given by:
    \begin{align*}
        \mathbb{E}[x] = 0, \quad \mathrm{Cov}[x] = (a e^{2b^2} + c)\mathbf{I}.
    \end{align*}
\end{lemma}

By applying Lemma \ref{le5}, the expectation and covariance of $\eta_{ij}(t)$ are obtained as:
\begin{align}\label{e33}
    \mathbb{E}[\eta_{ij}(t)]&=0,\nonumber\\
    \mathrm{Cov}[\eta_{ij}(t)]&=\left(\frac{z_i \|x_i(t) - x_j(t)\|^{\alpha}\mathrm{e}^{2(\frac{\ln 10}{20}\sigma_{dB})^2}}{C_0^2}\!+\!s_j\right)\mathbf{I}.
\end{align} Define $\mathcal{X}=\big[\mathcal{X}_{l_1},\mathcal{X}_{l_2},\ldots, \mathcal{X}_{l_{N_\mathcal{E}}},\eta_{{l_1}_s},\eta_{{l_2}_s},\ldots,\eta_{{l_{N_\mathcal{E}}}_s},\eta_{{l_1}_z},\eta_{{l_2}_z},$ $\ldots,\eta_{{l_{N_\mathcal{E}}}_z}\big]$ and the natural filtration $\mathcal{F'}_t=\sigma(\mathcal{X}(0),\ldots,\mathcal{X}(t))$. The term $\mathcal{F}'_t$ represents the $\sigma$-algebra generated by the random variables up to time $t$, representing the complete information history of the channel environment. The convergence properties of $\xi(t)$ for the case where $1<\alpha < 2$ are characterised as follows.
\begin{proposition}
Suppose that Assumptions \ref{ass1} and \ref{ass2} hold for the MAS described by \eqref{e1}. With $1<\alpha < 2$ in \eqref{e8}, if the penalty factor $c(t)$ in \eqref{e15} is chosen to be strictly positive and $c(t)\in \ell_1$, then under Algorithm 1, the tracking error $\xi(t)$ converges to $\xi^*$ in \eqref{e32} a.s. and in mean square. More specifically,
\begin{align}\label{e34}
     \mathbb{E}[\xi^*]&=\Phi(\infty,0)\xi(0),\nonumber\\
     \mathbb{E}[\|\xi^*\|^2] &\leq \exp\left( -2\lambda_{\min}^{\Psi} \sum\limits_{j=t_2}^\infty c(j)\right)\|\xi(0)\|^2\prod_{j=0}^{t_2-1} \|I - c(j)\Psi_j\|^2\nonumber\\&\quad+\rho^2_{K\mathcal{B}} \rho_\Phi^2\sum\limits_{j=0}^\infty c^2(j)\bigg[ nsN_\mathcal{E} +\frac{2^{\alpha-1}nz\mathrm{e}^{2(\frac{\ln 10}{20}\sigma_{dB})^2}}{C_0^2}\nonumber \\
    & \quad \times \left((d^\top d)^{\frac{\alpha}{2}}+{\rho'_{\xi}}^{\frac{\alpha}{2}}\right) \bigg],
\end{align}
where $\rho'_\xi$ is the upper bound of $\mathbb{E}[\xi^\top(t)\xi(t)]$ when $1<\alpha<2$.
\end{proposition}
\begin{proof}
The proof proceeds analogously to that of Theorem~\ref{th2} up to the derivation of \eqref{e26}. Therefore, detailed steps are omitted here. It is also worth noting that the pair $\{\bar{M}(t), \mathcal{F}'_t\}$ constitutes a martingale sequence. Considering the shadowing effect $\mathcal{X}_{ij}(t)$ and $1<\alpha < 2$, in conjunction with \eqref{e33},
\begin{align}\label{e35}
       & \mathbb{E}[\eta^\top(t)\eta(t)]
        \nonumber\\
        &= n \sum\limits_{i={1}}^{N}\sum\limits_{j\in\mathcal{N}_i}\mathbb{E}[\frac{z_{i} \|\xi_{ij}(t) - d_{ij}\|^{\alpha}\mathrm{e}^{2(\frac{\ln 10}{20}\sigma_{dB})^2}}{C_0^2}\!+\!s_j]\nonumber\\
        &\leq n \sum\limits_{i={1}}^{N}\sum\limits_{j\in\mathcal{N}_i}\mathbb{E}[\frac{2^{\alpha-1}\mathrm{e}^{2(\frac{\ln 10}{20}\sigma_{dB})^2}z_{i}(\|\xi_{ij}(t)\|^\alpha+\|d_{ij}\|^\alpha)}{C_0^2}+s_{j}]\nonumber\\
        &\leq \frac{2^{\alpha-1}nz\mathrm{e}^{2(\frac{\ln 10}{20}\sigma_{dB})^2}\left((d^\top d)^{\frac{\alpha}{2}}+\mathbb{E}[(\xi^\top(t)\xi(t))^{\frac{\alpha}{2}}]\right)}{C_0^2}\!+\!nsN_\mathcal{E}.
    \end{align} The second inequality holds since $\alpha > 1$, and the third follows from Jensen's inequality for concave functions $x^a$ with $a = \frac{\alpha}{2} < 1$. A derivation similar to (18)–(20) yields:
    \begin{align}\label{e36}
        &\mathbb{E}[\xi^\top(t\!+\!1)\xi(t\!+\!1)]\nonumber\\
        &\leq(1\!-\!2c(t)\lambda_{\min}^{\Psi}+c^2(t)\rho_{\Psi}^2)\mathbb{E}[\xi^\top(t)\xi(t)]
    \nonumber\\
    &\quad+\frac{2^{\alpha-1}nz\mathrm{e}^{2(\frac{\ln 10}{20}\sigma_{dB})^2}c^2(t)\rho_{K\mathcal{B}}^2}{C_0^2}( \mathbb{E}[\xi^\top(t)\xi(t)])^{\frac{\alpha}{2}}\nonumber\\
    &\quad+\frac{2^{\alpha-1}nz\mathrm{e}^{2(\frac{\ln 10}{20}\sigma_{dB})^2} c^2(t)\rho_{K\mathcal{B}}^2(d^\top d)^\frac{\alpha}{2}}{C_0^2}+nsN_\mathcal{E}c^2(t)\rho_{K\mathcal{B}}^2.
    \end{align} Let $V_t\triangleq\mathbb{E}[\xi^\top(t)\xi(t)]$, $\Delta V_t\triangleq V_{t+1}-V_t$ and
    \begin{align*}
        L_1&=2\lambda_{\min}^{\Psi},\ L_2=\rho_{\Psi}^2,\ L_3=\frac{2^{\alpha-1}nz\mathrm{e}^{2(\frac{\ln 10}{20}\sigma_{dB})^2}\rho_{K\mathcal{B}}^2}{C_0^2},\\
        L_4&=\frac{2^{\alpha-1}nz\mathrm{e}^{2(\frac{\ln 10}{20}\sigma_{dB})^2} \rho_{K\mathcal{B}}^2(d^\top d)^\frac{\alpha}{2}}{C_0^2}+nsN_\mathcal{E}\rho_{K\mathcal{B}}^2.
    \end{align*}
    It can be verified that $L_1,L_2,L_3,L_4>0$. The inequality in \eqref{e36} can be rewritten as:
    \begin{align}\label{e37}
        \Delta V_t &\leq -c(t)  \left[ (L_1 - L_2 c(t)) V_t - L_3 c(t) V_t^{\frac{\alpha}{2}} - L_4 c(t) \right] \nonumber \\ 
         &\leq -c(t)\left[ L_1V_t-c(t)(L_2V_t+L_3V_t^{\frac{\alpha}{2}}+L_4) \right]\nonumber\\
         &\leq-c(t)\left(L_2V_t+L_3V_t^{\frac{\alpha}{2}}+L_4\right)\nonumber\\
         &\quad \times\left[ \frac{L_1}{L_2+\frac{L_3}{V_t^{1-\frac{\alpha}{2}}}+\frac{L_4}{V_t}} -c(t) \right].
    \end{align} 
    Note that $1 - \frac{\alpha}{2} > 0$. Suppose that $\lim\limits_{t \to \infty} V_t = \infty$. Observing the fraction in \eqref{e37}, as $V_t \to \infty$, 
    $$\lim_{V_t \to \infty} \frac{L_1}{L_2 + \frac{L_3}{V_t^{1-\alpha/2}} + \frac{L_4}{V_t}} = \frac{L_1}{L_2}$$ 
    A sufficiently large constant $V_M >(\frac{\alpha L_3}{2L_2})^{\frac{1}{1-\frac{\alpha}{2}}}$ can be defined such that $V_M-\frac{L_3}{L_2}V_M^{\frac{\alpha}{2}}>\frac{L_4}{L_2}$. Consequently, whenever $V_t > V_M$, the fraction $L_1 \big( L_2 + L_3 V_t^{\alpha/2 - 1} + L_4 V_t^{-1} \big)^{-1}$ is strictly greater than $\frac{L_1}{2L_2}$. Furthermore, since the sequence $c(t)$ diminishes to $0$ over time as $c(t)\in \ell_1$, there exists a time instant $t_4 > 0$ such that for all $t \ge t_4$, the inequality $c(t) < \frac{L_1}{2L_2}$ holds. This, in turn, indicates that $\Delta V_t < 0$ for $t > t_4$, which contradicts the premise that $V_t$ is unbounded. Therefore, $V_t$ is restricted to be bounded. Let this bound be denoted by $\rho'_{\xi}$; that is, for any $t \ge 0$, $\mathbb{E}[\xi^\top(t)\xi(t)]\leq\rho'_{\xi}$. The remainder of the convergence proof and the analysis of the limit random vector $\xi^\ast$ proceed analogously to the proof of Theorem \ref{th2}. Therefore, they are omitted here for brevity. The proof is completed.
\end{proof}

\begin{remark}\label{re4}
While the analysis above establishes global convergence for $1<\alpha \leq 2$, a structural limitation arises when $\alpha > 2$. Observe that the term $V_t^{-(1-\alpha/2)}$ in inequality \eqref{e37} exhibits a positive exponent in this case, leading to unbounded growth in the denominator as $V_t$ increases. This behaviour precludes the satisfaction of the negative drift condition for arbitrarily large states. Thus, global convergence is lost. The system instead exhibits semi-global convergence, requiring the initial states $\xi^\top(0)\xi(0)$ to lie within a specific region of attraction where the stabilising drift outweighs the noise amplification. One feasible solution lies in the design of a state-dependent gain sequence $c(t)$, aiming to counteract the super-linear growth of the noise intensity. Investigating such adaptive protocols for $\alpha > 2$ is left for future work.
\end{remark}


\subsection{Privacy Analysis}

Recall that the protection target in this work is the weighting ratio $R_i^{-1}Q_i$ for all $i\in\mathcal{V}$. This metric is critical as it intrinsically reflects the operational priorities encoded within the cost function. Specifically, a larger ratio corresponds to an aggressive control policy, where minimising state deviation takes precedence. In a physical context, this signifies that the agent is engaged in high-precision tasks. Conversely, a smaller ratio indicates a conservative, energy-efficient strategy, typically adopted when an agent experiences power constraints or actuator faults. Consequently, the disclosure of this ratio exposes the internal operational status of the agent, enabling adversaries to identify vulnerable units or anticipate execution trajectories. 

To quantify the safeguard against the information leakage of $R_i^{-1}Q_i$ by Algorithm \ref{alg1}, a formal DP analysis is presented. Note that the effective noise intercepted by the eavesdropper attains its minimum variance at the transmitter. For analytical simplicity and to establish a worst-case privacy guarantee, this subsection assumes that the eavesdropper is exactly co-located with the transmitter, where the eavesdropper's effective noise model is given as
\begin{align*}
    \eta_{ij}^e(t)\sim\mathcal{N}(0,s_j\mathbf{I})
\end{align*}

The sensitivity parameter $\Delta_{K_i,t}$ is firstly defined as follow:
    \begin{align}\label{38}
        \Delta_{K_{i,t}}& \triangleq \max \Big\{  \|K_{i,t} - K'_{i,t}\| : 
 R_i^{-1}Q_i \in D_1,\nonumber \\&R_i'^{-1}Q'_i \in D'_1,
 D_1 \text{ and } D'_1 \text{ satisfy Definition \ref{de1}} \Big\}.
    \end{align}
    This metric quantifies the maximum deviation in the gain matrix $K_{i,t}$ induced by the adjacency of  weighting ratios in Definition \ref{de1}. To facilitate the analysis, step-wise parameters $(\epsilon_t, \delta_t)$ are introduced to quantify the privacy budget at each time instant $t$. Within this framework, $\delta_t$ characterises the relaxation probability of the differential privacy guarantee in \eqref{e12}, which is calibrated to navigate the inevitable trade-off between privacy security and system performance. 
\begin{theorem}\label{th3}
Suppose that Assumptions \ref{ass1} and \ref{ass2} hold,  and that the eavesdropper is co-located with the transmitter. Furthermore, let $c(t)>0$ and $\delta'_t\in(0,\frac{1}{2})$ be user-defined sequences. Then for any finite time $t_f>0$, Algorithm \ref{alg1} guarantees cumulative $(\epsilon,\delta)$-DP for the weighting ratios $R_i^{-1}Q_i,\ \forall i\in\mathcal{V}$, and
\begin{align*}
    \epsilon &\leq \frac{N_{\mathcal{E}}C_{\Delta}}{2\underline{s}} \sum_{t=0}^{t_f} c(t) + \frac{N_{\mathcal{E}}\sqrt{C_{\Delta}}}{\sqrt{\underline{s}}} \sum_{t=0}^{t_f} \sqrt{c(t)} \mathcal{Q}^{-1}(\delta'_t), \\[10pt]
\delta&\leq N_{\mathcal{E}} \sum_{t=0}^{t_f} \delta'_t + \frac{2\Delta_{K}^2N_{\mathcal{E}}^2(\rho'_{\xi} + \rho'_{\eta})}{C_{\Delta}} \sum_{t=0}^{t_f} c(t),
\end{align*}
where
\begin{align*}
    \rho'_\eta&=\frac{2^{\alpha-1}nz\mathrm{e}^{2(\frac{\ln 10}{20}\sigma_{dB})^2}\left((d^\top d)^{\frac{\alpha}{2}}+(\rho'_\xi)^{\frac{\alpha}{2}}\right)}{C_0^2}\!+\!nsN_\mathcal{E},\\
    \Delta_{K}&=\max\limits_{i,t}\{\Delta_{K_{i,t}}\},\quad \underline s=\min\limits_j \{s_j\},
\end{align*}%
$C_{\Delta}$ is a tunable parameter and selected to satisfy $C_\Delta>\sup\limits_{t\ge0}{2c(t)|\mathcal{N}_i| \Delta^2_{K}(\rho'_{\xi}+\rho'_\eta)}$. Furthermore, as  $t_f \rightarrow \infty$ 
select 
\begin{align}\label{con2}
    c(t)>0,\ c(t)\in\ \ell_1,\ \delta'_t\in\ell_1,\ \sqrt{c(t)} \mathcal{Q}^{-1}(\delta'_t)\in \ell_1,
\end{align} and sufficiently large $C_{\Delta}$, Algorithm \ref{alg1} can achieve a bounded cumulative privacy budget (i.e., $\epsilon < \infty$) and a valid cumulative failure probability (i.e., $\delta < 1$) over an infinite horizon.
\end{theorem}
\begin{proof}
 The privacy properties of the proposed LQR mechanism $\mathcal{M}_{t}$ at time $t$ are first analysed. Let $D$ and $D'$ be two adjacent databases as defined in Definition \ref{de1}. Noting that the weighting ratio $R_i^{-1}Q_i$ is queried exclusively during the execution of $\mathcal{M}_{t,i}$ for all $ i\in\mathcal{V}$, the subsequent analysis focuses on quantifying the privacy level of local $\mathcal{M}_{t,i}$. Furthermore, given that the gain $\|K_{i,t}\|$ is uniformly bounded for all $i \in \mathcal{V}$ and $t \geq 0$, the boundedness of $\Delta_{K_{i,t}}$ and existence of $\triangle_K$ are immediately established. Let $\Delta_{K_i}$ denote the uniform upper bound, such that $ \Delta_{K_i} \ge \Delta_{K_{i,t}} $. Based on Lemma \ref{le2}, the sensitivity of the output $x_i(t+1)$ with respect to $R_i^{-1}Q_i$ is derived as
    \begin{align}\label{EX5}
        \Delta_i(t)&=\max\limits_{D_1,D_1'} \|\mathcal{M}_{t,i}(D_1)-\mathcal{M}_{t,i}(D_1')\| \nonumber\\
        &=\max\limits_{D_1,D_1'} \|x_i(t+1)\mid_{D_1}-x_i(t+1)\mid_{D_1'}\| \nonumber\\
        &=\max\limits_{D_1,D_1'} \|c(t)\left(K_{i,t}\mid_{D_1}-K_{i,t}\mid_{D_1'}\right)\nonumber\\
        &\quad \times\sum\limits_{j\in\mathcal{N}_i}a_{ij}  (\hat{x}_{ij}(t)+d_{ij}-x_i(t))\| \nonumber\\
        &\leq c(t)\Delta_{K_{i,t}}\|\sum\limits_{j\in\mathcal{N}_i}a_{ij}  (\hat{x}_{ij}(t)+d_{ij}-x_i(t))\|.
    \end{align}
where $\hat{x}_{ij}=x_j+\eta_{ij}$ and the third equality follows from \eqref{e17} and the Adaptive Sequential Composition Theorem \cite{dwork2014algorithmic}. As revealed by this theorem, the term $c(t)K_{i,t}\sum_{j\in\mathcal{N}_i}a_{ij} (\hat{x}_{ij}(t)-x_i(t))$ is determined by the output of the preceding mechanism $\mathbf{\mathcal{M}}_{t-1,i}$. Consequently, it is treated as a fixed constant at time $t$ and remains strictly invariant with respect to the shift between databases $D$ and $D'$. 

Nevertheless, given the stochastic nature of $x_i(t)$ and $x_j(t)$ as the outputs of preceding mechanisms, alongside the Gaussian noises present in $\hat{x}_{ij}$, it is impossible to establish a deterministic upper bound for $\Delta_i^2(t)$. A more tractable approach is to evaluate the tail probability $\delta''_{i,t}\in(0,1)$ that $\Delta_i^2(t)$ exceeds a specified upper bound $c(t)C_\Delta>0$. Based on \eqref{EX5},
\begin{align*}
    \Delta_i^2(t)&\leq c^2(t)\Delta^2_{K_{i,t}}\|\sum\limits_{j\in\mathcal{N}_i}a_{ij}  (\hat{x}_{ij}(t)+d_{ij}-x_i(t))\|^2 \nonumber \\
    &= c^2(t)\Delta^2_{K_{i,t}}\|\sum\limits_{j\in\mathcal{N}_i}a_{ij}\eta_{ij}(t)-\sum\limits_{j\in\mathcal{N}_i}a_{ij}\xi_{ij}(t)\|^2 \nonumber\\
    &\leq  c^2(t)|\mathcal{N}_i| \Delta^2_{K_{i,t}}\sum_{j\in\mathcal{N}_i} a_{ij}^2 \|\eta_{ij}(t)-\xi_{ij}(t) \|^2 \nonumber\\
    &\leq 2c^2(t)|\mathcal{N}_i| \Delta^2_{K_{i,t}}\left(\sum_{j\in\mathcal{N}_i} \|\xi_{ij}(t)\|^2+\sum_{j\in\mathcal{N}_i} \|\eta_{ij}(t)\|^2\right) \nonumber\\
    &\leq 2c^2(t)|\mathcal{N}_i| \Delta^2_{K_{i,t}}(| \|\xi(t)\|^2+| \|\eta(t)\|^2).
\end{align*} By applying Markov's inequality ($\Pr(X\ge a)\le\frac{\mathbb{E}(X)}{a}$) \cite{ash2014real},
\begin{align}\label{EX6}
    \Pr(\Delta_i^2(t)\ge c(t)C_\Delta)&\le \frac{2c^2(t)|\mathcal{N}_i| \Delta^2_{K_{i,t}}\mathbb{E}(| \|\xi(t)\|^2+| \|\eta(t)\|^2)}{c(t)C_\Delta} )\nonumber\\
    &\le \frac{2c(t)|\mathcal{N}_i| \Delta^2_{K_{i,t}}(\rho'_{\xi}+\rho'_\eta)}{C_\Delta},
\end{align}where $\mathbb{E}[\eta^\top(t)\eta(t)]\leq\rho'_\eta$. It is obtained $\delta''_{i,t}\leq\frac{2c(t)|\mathcal{N}_i| \Delta^2_{K_{i,t}}(\rho'_{\xi}+\rho'_\eta)}{C_\Delta}$. In accordance with \eqref{EX6}, and given the definiteness of $\Delta_i(t)$,
\begin{align}\label{EX7}
     \Pr(\Delta_i(t)\ge \sqrt{c(t)C_\Delta})\le \frac{2c(t)|\mathcal{N}_i| \Delta^2_{K_{i,t}}(\rho'_{\xi}+\rho'_\eta)}{C_\Delta}.
\end{align}

Given that the noise on each link $(i,j)$ is independent and heterogeneously distributed, the privacy guarantees exhibited by the $\mathcal{M}_{t,i}=\{\mathcal{M}_{t,ij}\mid j\in\mathcal{N}_i\},\ \mathcal{M}_{t,ij}=x_i(t+1)+\eta_{ij}^e(t+1)$ is link-dependent. Since the privacy relaxation probability $\delta'_t$ associated with the $\epsilon$ in Lemma \ref{le2} is a user-defined parameter, it is set as $\delta'_{ij,t}=\delta'_{i,t}=\delta'_{t}$. Following Lemma \ref{le2}, the privacy budget $\epsilon_{ij,t}$ for $\mathcal{M}_{t,ij}$ satisfies:
    \begin{align}\label{e39}
\epsilon_{ij,t}\! 
&\leq \frac{\Delta_{i}^2(t)}{2s_j}+\frac{\Delta_{i}(t)\mathcal{Q}^{-1}(\delta'_{t})}{\sqrt{s_j}}. 
    \end{align}
Combining \eqref{EX6}, \eqref{EX7} and \eqref{e39} yields $\Pr\left(\epsilon_{ij,t}  \ge C_\epsilon(t) \right)\leq  \frac{2c(t)|\mathcal{N}_i| \Delta^2_{K_{i,t}}(\rho'_{\xi}+\rho'_\eta)}{C_\Delta},\ C_\epsilon(t)  =\frac{{c(t)C_\Delta}}{2s_j} \! + \!\frac{\sqrt{c(t)C_\Delta}\mathcal{Q}^{-1}(\delta'_{t})}{\sqrt{s_j}}.$ Recalling Definition \ref{de3}, and noting that $\delta$ in inequality \eqref{e12} represents the probability of violating the privacy guarantee, it is concluded 
\begin{align}\label{EX8}
    \epsilon_{ij,t}=C_\epsilon(t),\quad \delta_{ij,t}=\delta'_t+\delta''_{i,t}.
\end{align}
    
    It is noted that the release of $\mathcal{M}_{t,ij},\ \forall j\in\mathcal{N}_i$ can be viewed as a sequence of independent mechanisms applied to the same underlying state $x_i(t+1)$. Consequently, by virtue of the sequential composition theorem \cite{dwork2014algorithmic}, the aggregate privacy budget for agent $i$'s weighting ratio $R_i^{-1}Q_i$ satisfies $\epsilon_{i,t}=\sum\limits_{j\in \mathcal{N}_i}\epsilon_{ij,t},\ \delta_{i,t}=\sum\limits_{j\in \mathcal{N}_i}\delta_{ij,t}$. By incorporating \eqref{EX8},
    \begin{align}\label{e40}
        \epsilon_{i,t}&=\sum\limits_{j\in \mathcal{N}_i}C_\epsilon(t)=\sum\limits_{j\in \mathcal{N}_i}\left(\frac{{c(t)C_\Delta}}{2s_j} \! + \!\frac{\sqrt{c(t)C_\Delta}\mathcal{Q}^{-1}(\delta'_{t})}{\sqrt{s_j}}\right),\nonumber\\
        \delta_{i,t}&\leq\sum\limits_{j\in \mathcal{N}_i}\left(\delta'_t+\frac{2c(t)|\mathcal{N}_i| \Delta^2_{K_{i,t}}(\rho'_{\xi}+\rho'_\eta)}{C_\Delta}\right).
    \end{align}
Note that $\epsilon_{i,t}$ represents the privacy budget and $\delta_{t}$ the failure probability of $\mathcal{M}_{t,i}$. By the parallel composition theorem \cite{dwork2014algorithmic}, the global privacy parameters for $\mathcal{M}_t$ are given by $\epsilon_{t} = \max\limits_{i} \{ \epsilon_{i,t} \}$ and $\delta_{t} = \max\limits_{i} \{ \delta_{i,t} \}$. Following \eqref{e40}, the following result is obtained:
\begin{align}\label{e42}
   \epsilon_{t}&\leq \frac{{c(t)N_{\mathcal{E}}C_\Delta}}{2\underline s}+\frac{\sqrt{c(t)C_\Delta}N_{\mathcal{E}}\mathcal{Q}^{-1}(\delta'_{t})}{\sqrt{\underline s}},\nonumber\\
   \delta_t&\leq N_{\mathcal{E}}\delta'_t+\frac{2c(t)N_{\mathcal{E}}^2 \Delta^2_{K_{t}}(\rho'_{\xi}+\rho'_\eta)}{C_\Delta},
\end{align} where $\Delta_{K_t}=\max\limits_{i}\{\Delta_{K_{i,t}}\}$. 

Note that $(\epsilon_t,\delta_t)$-DP is satisfied by the mechanism $\mathcal{M}_t$ at each time step~$t$, where $\epsilon_t$ and $\delta_t$ satisfy \eqref{e42}. By invoking the adaptive sequential composition theorem~\cite{dwork2014algorithmic} to account for the adaptive dependence on the execution history, the cumulative privacy parameters $\epsilon=\sum\limits_{t=0}^{t_f}\epsilon_t$ and $\delta=\sum\limits_{t=0}^{t_f}\delta_t$ are derived as follows:
\begin{align}\label{EX9}
    \epsilon &\leq \frac{N_{\mathcal{E}}C_{\Delta}}{2\underline{s}} \sum_{t=0}^{t_f} c(t) + \frac{N_{\mathcal{E}}\sqrt{C_{\Delta}}}{\sqrt{\underline{s}}} \sum_{t=0}^{t_f} \sqrt{c(t)} \mathcal{Q}^{-1}(\delta'_t), \nonumber\\
\delta&\leq N_{\mathcal{E}} \sum_{t=0}^{t_f} \delta'_t + \frac{2\Delta_{K}^2N_{\mathcal{E}}^2(\rho'_{\xi} + \rho'_{\eta})}{C_{\Delta}} \sum_{t=0}^{t_f} c(t),
\end{align}where $\Delta_{K}=\max\limits_{t}\{\Delta_{K_{t}}\}$. Because $ c(t)\in\ \ell_1,\ \delta'_t\in\ell_1$ and $\ \sqrt{c(t)} \mathcal{Q}^{-1}(\delta'_t)\in \ell_1$ in \eqref{con2}, it is immediate to show $\epsilon<\infty$ and $\delta<1,$\footnote{According to \eqref{EX9}, $\delta < 1$ can always be guaranteed by appropriately selecting the constant coefficients in sequences $c(t)$ and $\delta'_t$.} as $t_f\rightarrow\infty$. The proof is completed.
\end{proof}

Although Theorem \ref{th3} establishes that Algorithm \ref{alg1} achieves a strictly bounded cumulative DP guarantee under \eqref{con2}, comparing this result with the reference state protection in \cite{ma2026inherent} reveals a distinct quantitative disparity. Specifically, whilst both targets maintain theoretically bounded privacy over an infinite horizon, masking the weighting ratio inherently incurs a larger privacy loss. This demonstrates that inherent communication noise, despite being a valuable natural privacy mechanism, provides non-uniform protection. Its efficacy is governed by how the protected information structurally couples with the cooperative dynamics. This inherent characteristic also motivates a new research paradigm, wherein natural physical randomness and actively structured artificial noise are synergistically integrated to satisfy diverse and stringent privacy requirements.

Combining the requirement for unbiased convergence in Corollary \ref{co1}, $c(t)\in\ell_2, \notin \ell_1$, with the conditions in \eqref{con2} for achieving a bounded cumulative privacy budget in Theorem \ref{th3}, it is clear that exact asymptotically unbiased convergence and a bounded DP guarantee can not be achieved simultaneously over an infinite horizon. This result corroborates the dilemma revealed in \cite{huang2015differentially,cortes2016differential,nozari2015differentially,nozari2017differentially,huang2024differential}. Unlike artificial noise injection, where noise variance can be scaled up to suppress diverging sensitivity over infinite iterations in\cite{ma2025distributed}, inherent physical noise is limited by natural channel properties. Therefore, the only method to restrict cumulative privacy loss is to force the penalty factor $c(t)$ to decay rapidly (i.e., $c(t) \in \ell_1$). However, this early drop in control effort deprives the multi-agent system of the sustained power needed to overcome random disturbances, sacrificing unbiased convergence. This highlights the basic conflict between the continuous control needed for exact coordination and the reduced information leakage needed for strict privacy.

\begin{remark}
    Utilising the asymptotic property $\mathcal{Q}^{-1}(\delta_t)=O(\sqrt{\ln(1/\delta_t)})$, the selection of $c(t)$ and $\delta_t$ in  \eqref{con2} can be significantly simplified. Admittedly, these conditions are conservative. This conservatism stems primarily from two factors: the extensive algebraic relaxations employed to derive upper bounds, and the stringent requirement of bounded cumulative privacy budget $\epsilon$ over an infinite horizon. This requirement implicitly assumes the existence of an eavesdropper who continuously colludes with adversarial agents from the system's inception, an assumption that is inherently unrealistic. Furthermore, note that privacy leakage is most severe when the eavesdropper is co-located with the transmitter, as the communication noise variance reaches its minimum at this position. This worst-case scenario further constrains the feasible range of $c(t)$ and $\delta_t$.
\end{remark}

\begin{remark}
  It is worth noting that while the marginal distribution of the communication noise exhibits complex non-Gaussian mixture characteristics, the conditional distribution given the channel history remains Gaussian. This property ensures the validity of applying Lemma \ref{le2} at each time step. Another key aspect of proving Theorem \ref{th3} is addressing the state-dependent sensitivity $\Delta_i(t)$. Because the system state evolves within an unbounded continuous space ($\mathbb{R}^n$), the conventional worst-case global sensitivity inevitably evaluates to infinity, which renders standard DP mechanisms infeasible. To overcome this barrier, by leveraging the bounded mean-square tracking error and applying Markov's inequality, a probabilistic boundary is established, restricting the local sensitivity within a finite truncation threshold $C_\Delta$ with high probability. The tail probability $\delta''_{i,t}$ wherein the sensitivity exceeds the $C_\Delta$ is then rigorously accommodated by absorbing it into the overall privacy relaxation parameter $\delta_{ij,t}$ via the union bound. In essence, this establishes a formal link between the system's statistical stability and its privacy guarantees.
\end{remark}

\subsection{Privacy Efficiency: Weighting Ratios vs. Gradients}

In this subsection, a sensitivity analysis demonstrates that, under the addition/removal adjacency definition, adjacency relations based on weighting ratios yield less conservative privacy bounds than those based on gradients. This implies that gradients inherently contain redundant information, such as the underlying cost structure and dynamic system states, rendering the privacy budget $(\epsilon,\delta)$ calculated for gradient protection overly conservative. Consequently, this justifies the rationale of this paper to focus exclusively on the protection of weighting ratios. The analytical procedure is presented as follows.

Rather than obtaining the optimal input $u_i(t)$ via \eqref{e16}--\eqref{e17}, an alternative method relies on numerical optimisation approaches (e.g., see \cite[Chapter 8]{borrelli2017predictive}). To illustrate how gradient information directs the search for the optimal input, a gradient descent method is adopted here. Specifically, by defining the predicted input sequence over the finite horizon $T$ as $U_i(t)=\left[\begin{matrix} \bar{u}_i(t|t)^\top & \bar{u}_i(t+1|t)^\top &\ldots &\bar{u}_i(t+T-1|t)^\top \end{matrix}\right]^\top$, the analytical gradient of \eqref{eq:14a} with respect to the input sequence is given by
\begin{align}\label{e19}
    \nabla_{U_i(t)} J_i &= 2 \bar{R}_i U_i(t) + 2 c(t)\sum\limits_{j\in\mathcal{N}_i} a_{ij} \Phi^\top \bar{Q}_i \big(\Phi U_i(t) \nonumber\\
    &\quad + \Lambda[ x_i(t)-d_{ij}-\hat{x}_{ij}(t)]\big),
\end{align}
where $$\nabla_{U_i(t)} J_i\!=\!\left[\begin{matrix} \nabla_{\bar{u}_i(t|t)}J_i^\top \!& \!\nabla_{\bar{u}_i(t+1|t)}J_i^\top \!&\!\ldots \!&\!\nabla_{\bar{u}_i(t+T-1|t)}J_i^\top \end{matrix}\right]^\top.$$ Let $k\ge0$ denote the internal iteration index within the optimisation solver at time $t$. The gradient descent update law is formulated as
\begin{align}\label{EX1}
    U_i^{(k+1)}(t) = U_i^{(k)}(t) - \alpha \nabla_{U_i(t)} J_i\big(U_i^{(k)}(t)\big),
\end{align}
where $\alpha > 0$ is an appropriately chosen step size to ensure convergence. To accelerate this iterative process and exploit temporal coherence, the initial sequence $U_i^{(0)}(t)$ is constructed by shifting the converged optimal sequence obtained at the previous time step $t-1$:
\begin{align*}
    U_i^{(0)}(t) = \left[\begin{matrix} u_i^*(t|t-1)^\top& \dots& u_i^*(\tau|t-1)^\top& u_i^*(\tau|t-1)^\top \end{matrix}\right]^\top.
\end{align*} where $\tau=t+T-2$. Iterating \eqref{EX1} drives the sequence $U_i^{(k)}(t)$ to converge to the optimal solution
\begin{align}\label{EXX1}
    U_i^{\ast}(t)=\left[\begin{matrix} u_i^{\ast}(t|t)^\top & u_i^{\ast}(t+1|t)^\top &\ldots &u_i^{\ast}(t+T-1|t)^\top \end{matrix}\right]^\top,
\end{align}
which satisfies the first-order optimality condition $$\left. \nabla_{U_i(t)} J_i \right|_{U_i(t) = U_i^{\ast}(t)} = \mathbf{0}.$$ In accordance with the receding horizon principle, only the first element of $U_i^{\ast}(t)$, namely $u_i^{\ast}(t|t)$ is applied as the control input at time $t$. Furthermore, it can be verified that $u_i^{\ast}(t|t)$ equals the input in \eqref{e17}. Within such a numerical optimisation approach, a sensitive dataset is formulated from the gradients $\{\nabla_{U_i^{(0)}(t)} J_i\}_{i=0}^N$, upon which the adjacency relation is built. The rationale for targeting this specific iteration is that, when evaluated at this unoptimised warm-start state, the initial gradient $\nabla_{U_i^{(0)}(t)}J_i$ remains non-zero and explicitly manifests the raw, unmasked interaction between the tracking accuracy preference $Q_i$ and the control effort penalty $R_i$.

To facilitate comparison, an alternative adjacency relation definition \cite{dwork2014algorithmic} is adopted for the upcoming theorem. 

\begin{definition}[Addition/Removal-based adjacency, \cite{dwork2014algorithmic}]\label{de4}
    Let $\mathcal{D}_i$ denote the universe of admissible sets. Two datasets $D_i, D_i' \subseteq \mathcal{D}_i$ are adjacent if and only if they differ by the addition or removal of a single element, i.e., $\||D_i| - |D_i'|\| = 1$.
\end{definition}

The adoption of Definition \ref{de4} ensures that an eavesdropper cannot deduce whether a specific value within the sensitive dataset was actually utilised. This indistinguishability precisely addresses the requirement of protecting an agent's specific gradient or weighting ratio. Furthermore, evaluating these two metrics under Definition \ref{de1} relies on an adjacency threshold, $\theta_1$, which relies on the numerical scale of the target variables. Since the weighting ratio and the gradient possess different natural scales of variation, imposing an identical threshold on both would inevitably skew the comparison. By focusing purely on addition/removal of a element, Definition \ref{de4} can circumvent this issue, yielding a scale-neutral metric.

\begin{theorem}\label{th1}
Consider the LQR formulation presented in \eqref{e14} and adjacency in Definition \ref{de4}. With an identical privacy budget $(\epsilon, \delta)$, the noise variance required to protect the gradients $ \nabla_{U_i^{(0)}(t)} J_i\neq0$ is no less than that required for the weighting ratios $R_i^{-1}Q_i$.
\end{theorem}
\begin{proof}
   A formal comparison is first established between the ratio matrix $R_i^{-1}Q_i$ and the initial element of the gradient vector $\nabla_{U_i^{(0)}(t)} J_i$, denoted as $\nabla_{u_i^*(t|t-1)}J_i$. Define the fundamental sensitive dataset drawn from the feasible dataset as $D^0_i \triangleq \{(Q_i, R_i) \mid Q_i \in \bar{\mathcal{Q}}_i, R_i \in \bar{\mathcal{R}}_i\}$, where $\forall\bar{\mathcal{Q}}_i\subset{\mathcal{Q}}_i,\forall\bar{\mathcal{R}}_i\subset{\mathcal{R}}_i$. By introducing the ratio matrix $S_i=R_i^{-1}Q_i$, a transformed dataset $D^1_i \triangleq \{(S_i, R_i) \mid S_i \triangleq R_i^{-1}Q_i, Q_i \in \bar{\mathcal{Q}}_i, R_i \in \bar{\mathcal{R}}_i\}$ is constructed. It is noted that, given the invertibility of $R_i$, a surjection is maintained between the original space $D^0_i$ and the transformed space $D^1_i$. Define the ratio sensitive set as $D^2_i \triangleq \{ R_i^{-1}Q_i \mid Q_i \in \bar{\mathcal{Q}}_i, R_i \in \bar{\mathcal{R}}_i \}$, and the gradient sensitive set as $D^3_i \triangleq \left\{ \nabla_{u_i^*(t|t-1)}J_i \mathrel{\bigg|} Q_i \in \bar{\mathcal{Q}}_i, R_i \in \bar{\mathcal{R}}_i \right\}$. The datasets $D^2_i$ and $D^{2'}_i$ are considered adjacent if they satisfies Definition \ref{de4}. Similarly, the corresponding adjacent sets for the gradient are defined as $D^3_i$ and $D^{3''}_i$. Suppose that $\overline{R_{i}^{-1}Q_i}$ and $\underline{\nabla_{u_i^*(t|t-1)}J_i}$ denote the elements added or deleted from $D^2_i$ and $D^3_i$, respectively. Subsequently, the influence of these elements on the base sensitive dataset $D^1_i$ is analysed, followed by an assessment of the impact on the transmitted state $x_i(t+1)$. 

  Given the symmetry of the adjacency relation and the arbitrary nature of the dataset $D_i^0$, it is sufficient to restrict the analysis to the operation of removing a single element. Consequently, $D^{2'}_i \triangleq \{ R_i^{-1}Q_i \mid  R_i^{-1}Q_i\neq \overline{R_{i}^{-1}Q_i}, Q_i \in \bar{\mathcal{Q}}_i, R_i \in \bar{\mathcal{R}}_i \}$, the corresponding mapping back to the original space $D^1_i$ is represented as $D^{1'}_i \triangleq \{(S_i,R_i) \mid  S_i\neq \overline{R_{i}^{-1}Q_i}, Q_i \in \bar{\mathcal{Q}}_i, R_i \in \bar{\mathcal{R}}_i \}$. Referring to \eqref{e19},
   \begin{align}\label{EX2}
\nabla_{U_i(t)} J_i &= 2 \bar{R}_i \bigg[ U_i(t) + c(t) \sum_{j \in \mathcal{N}_i} a_{ij}\bar{R}_i^{-1} \Phi^\top \bar{Q}_i \Big( \Phi U_i(t) \nonumber \\
&\quad + \Lambda \big( x_i(t) - d_{ij} - \hat{x}_{ij}(t) \big)  \Big)\bigg]\nonumber\\
&= 2 \bar{R}_i \bigg[ U_i(t) + c(t) \sum_{j \in \mathcal{N}_i} a_{ij} \Phi^\top \bar{R}_i^{-1}\bar{Q}_i \Big( \Phi U_i(t) \nonumber \\
&\quad + \Lambda \big( x_i(t) - d_{ij} - \hat{x}_{ij}(t) \big)  \Big)\bigg].
\end{align} 
   By treating \eqref{EX2} as a vector equation, the gradient $\nabla_{U_i(t)} J_i$ serves as the dependent variable, whilst $\bar{R}_i$ and $\bar{S}_i\triangleq\bar{R}_i^{-1}\bar{Q}_i$ act as the independent variables. Analysing the unknown parameters reveals that this mapping constitutes an underdetermined system: the total number of independent variables from $\bar{R}_i$ and $\bar{S}_i$ strictly exceeds the scalar equations provided by $\nabla_{U_i(t)} J_i$. Consequently, for a fixed $\nabla_{U_i(t)} J_i$, there are infinitely many combinations of $({S}_i, R_i)$ that map to the identical $\nabla_{U_i(t)} J_i$. The corresponding mapping of $D^{3''}_i$ back to the original space $D^1_i$ is $D^{1''}_i = \{(S_i,R_i) \mid  \nabla_{\bar{u}_i(t|t)}J_i^\top(S_i,R_i)\neq\underline{\nabla_{u_i^*(t|t-1)}J_i},\, Q_i \in \bar{\mathcal{Q}}_i,\, R_i \in \bar{\mathcal{R}}_i \}$. Under this mapping, the adjacencies $(D^2_i, D^{2'}_i)$ and $(D^3_i, D^{3'}_i)$ are transformed into the relations $(D^{1}_i, D^{1'}_i)$ and $(D^{1}_i, D^{1''}_i)$, respectively. With a dimensionality of $n=1$ and $Q, R \in \{1, 1.1, 1.2, \dots, 4\}$, $D^{1'}_i$ and $D^{1''}_i$ represent datasets obtained by deleting an arbitrary element from $D^{2'}_i$ and $D^{3'}_i$, respectively. The relationships between $D^{1}_i$, $D^{1'}_i$ and $D^{1''}_i$ are illustrated in Fig. \ref{fig1}.
\begin{figure}[htp]
    \centering
    \includegraphics[width=0.6\linewidth]{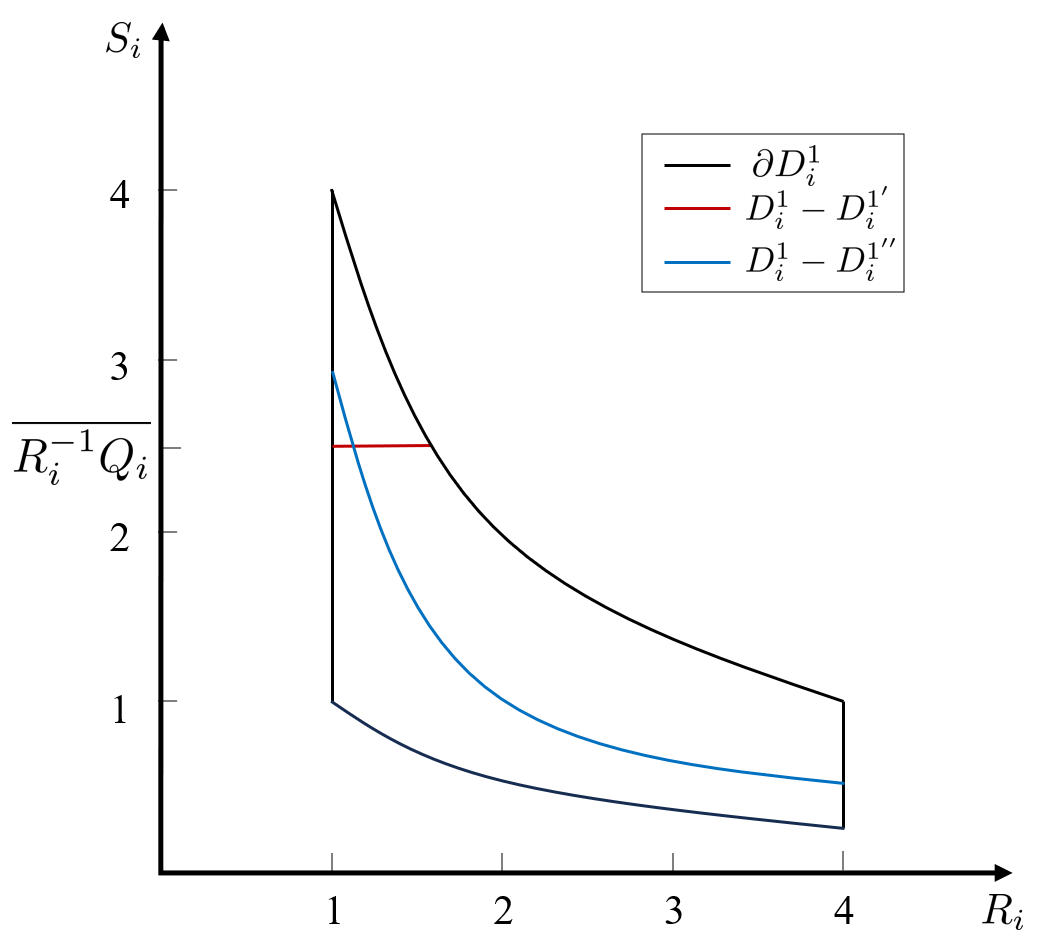}
    \caption{Illustration of the adjacency relations.}
    \label{fig1}
\end{figure}

 Recalling \eqref{e16},
   \begin{align}\label{EXX2}
&\left(\bar{R}_i+c(t)\sum\limits_{j\in\mathcal{N}_i}a_{ij}\Phi^\top\bar{Q}_i\Phi\right)^{-1}\Phi^\top\bar{Q}_i\Lambda\nonumber\\
&=\left(\mathbf{I}+c(t)\sum\limits_{j\in\mathcal{N}_i}a_{ij}\bar{R}_i^{-1}\Phi^\top\bar{Q}_i\Phi\right)^{-1}\bar{R}_i^{-1}\Phi^\top\bar{Q}_i\Lambda\nonumber\\
&=\left(\mathbf{I}+c(t)\sum\limits_{j\in\mathcal{N}_i}a_{ij}\Phi^\top\bar{R}_i^{-1}\bar{Q}_i\Phi\right)^{-1}\Phi^\top\bar{R}_i^{-1}\bar{Q}_i\Lambda.
   \end{align} It is concluded that the control gain $K_{i,t}$ in \eqref{e17} is uniquely determined by the ratio $S_i=R_i^{-1}Q_i$, rather than by the individual matrices $Q_i$ or $R_i$ independently. 
   To quantify the privacy implications of removal within the datasets $D_i^2$ and $D_i^3$, the resultant geometric perturbation in the output space is analysed. Let $U_{t,i}({D})$ denote the image set of the optimal control inputs generated under a specific parameter subset ${D}$, defined as $U_{t,i}({D}) \triangleq \{ u_i(t)|_{S_i} : S_i \in {D} \}$ and $u_i(t)$ is calculated by \eqref{e17}. The effective sensitivity can thus be formulated as the directed Hausdorff distance between the original image set and the reduced image set, measuring the maximum geometric 'void' created by the removal of a sensitive element. Consequently, the sensitivities $\Delta_{D^2_i}$ and $\Delta_{D^3_i}$, corresponding to the weighting ratios $R_i^{-1}Q_i$ and the gradients $\nabla_{u_i^*(t|t-1)}J_i$ respectively, are calculated as follows.
    \begin{align}\label{EX3}
     \Delta_{D^2_i}(t)&=\max\limits_{D^2_i,D^{2'}_i} \|\mathcal{M}_{t,i}(D^2_i)-\mathcal{M}_{t,i}(D^{2'}_i)\| \nonumber\\
&= \max_{D_i^1, D_i^{1'}} \rm{d_H}(U_{t,i}(D_i^1), U_{t,i}(D_i^{1'})) \nonumber \\
&= \max_{D_i^1, D_i^{1'}} \sup_{x \in U_{t,i}(D_i^1)} \inf_{y \in U_{t,i}(D_i^{1'})} ||x - y||.
 \end{align}
 Accordingly, it is obtained that
\begin{align}\label{EX4}
     \Delta_{D^3_i}(t)&=\max\limits_{D^3_i,D^{3'}_i} \|\mathcal{M}_{t,i}(D^3_i)-\mathcal{M}_{t,i}(D^{3'}_i)\| \nonumber\\
 &= \max_{D_i^1, D_i^{1''}} \rm{d_H}(U^\ast_{t,i}(D_i^1), U^\ast_{t,i}(D_i^{1''}))\nonumber \\
&= \max_{D_i^1, D_i^{1''}} \sup_{x \in U^\ast_{t,i}(D_i^1)} \inf_{y \in U^\ast_{t,i}(D_i^{1''})} ||x - y||\nonumber \\
&= \max_{D_i^1, D_i^{1''}} \sup_{x \in U_{t,i}(D_i^1)} \inf_{y \in U_{t,i}(D_i^{1''})} ||x - y||
 \end{align}
where $U^\ast_{t,i}({D}) \triangleq \{ u^\ast_i(t|t)|_{S_i} : S_i \in {D} \}$ and $u_i(t)$ is obtained by \eqref{EXX1}. The third equation holds since $u_i^{\ast}(t|t)=u_i(t)$ for same pair $(S_i,R_i)$. Combining \eqref{e17} and \eqref{EXX2} reveals that the optimal input $u_i(t)$ is determined uniquely by $S_i$. Given that the gradient $\nabla_{\overline{u}_i(t|t)}J_i$ is the image of $(R_i^{-1}Q_i, R_i)$ under a many-to-one mapping, its removal eliminates an entire equivalence class of ratios containing this maximising ratio. Consequently, the remaining set of ratios under gradient removal is a subset of those remaining under single-ratio removal, i.e., $\{S_i | S_i \in D_i^{1''}\} \subseteq \{S_i | S_i \in D_i^{1'}\}$ holds. For any given $x\in u_{i}(t)\mid_{ D^1_i}$, it implies 
\begin{align*}
    \inf\limits_{y\in u_{i}(t)\mid_{ D^{1''}_i}} \|x-y\|\ge\inf\limits_{y\in u_{i}(t)\mid_{ D^{1'}_i}} \|x-y\|.
\end{align*}
It is obtained $\Delta_{D^3_i}(t)\ge\Delta_{D^2_i}(t),\forall t\ge0$ \footnote{Since the many-to-one mapping implies that removing a single gradient element to construct $D_i^{1''}$ effectively removes an entire equivalence class of $(S_i, R_i)$ pairs. By selecting the ratio $\overline{R_i^{-1}Q_i}$ that maximises the sensitivity $\Delta_{D_i^2}(t)$ to define the single-ratio deletion set $D_i^{1'}$ above, and associating it with the removal of its corresponding gradient to define $D_i^{1''}$, it is established that $\Delta_{D^3_i}(t)\ge\Delta_{D^2_i}(t),\forall t\ge0$.}. This indicates that the output exhibits greater sensitivity to variations within the gradient set $D^3_i$. 

Based on Lemma \ref{le2}, it is established that protecting the predictive gradient $\nabla_{u_i^*(t|t-1)}J_i$ necessitates a larger noise variance than protecting the ratio matrix $R_i^{-1}Q_i$. More specifically, it is noted that $\nabla_{u_i^*(t|t-1)}J_i$ constitutes merely a single component within the full gradient vector $\left. \nabla_{U_i(t)} J_i \right|_{U_i(t) = U_i^{(0)}(t)}$. It can be directly deduced that, under identical privacy requirements, the noise variance required to safeguard this individual component is strictly smaller than that required for the entire vector. Through this transitive relationship, it is concluded that protecting the full gradient vector necessitates a considerably larger noise variance than protecting the ratio. The proof is completed.
\end{proof}

\begin{remark}\label{rre1}
     It should be noted that the proof of Theorem \ref{th1} establishes the subset inclusion $\{S_i \mid S_i \in D_i^{1''}\} \subseteq \{S_i \mid S_i \in D_i^{1'}\}$. The condition under which this relation degenerates into a strict subset ($\subsetneq$) remain an open question for further investigation. It is foreseeable that this condition might depend intricately on the boundary geometries of the feasible domains $\mathcal{Q}_i$ and $\mathcal{R}_i$, the transient system states, and the specific control inputs.
     
     The essence of this theorem lies not merely in geometric comparison, but rather in its exploration of the fundamental information boundaries within privacy-preserving control. Weighting ratios denote the irreducible core of an agent's intention, thus defining the essential information boundary necessary for control strategies. The gradient, on the other hand, is an operational representation. It may occasionally reduce to this boundary, but it can also expand far beyond it under certain conditions. Theorem \ref{th1} formalises this relationship to highlight a critical paradigm shift: rigorous DP design should not indiscriminately obscure operational outputs. Instead, designers should first explicitly identify and isolate the fundamental information boundaries that genuinely require protection.
\end{remark}

\section{Simulation Results}

To further illustrate the obtained results in Section \uppercase\expandafter{\romannumeral3}, three numerical simulations are conducted in this section. To accurately capture the impact of real-world communication noise, this part employs the noise model parameters experimentally obtained in \cite[\uppercase\expandafter{\romannumeral4}.B]{kunisch2008wideband}. Given the communication setup detailed in \cite[\uppercase\expandafter{\romannumeral4}.B]{kunisch2008wideband}, this paper considers a $4$-robot system operating in a $2$-dimensional space. In an environment with obstacles, the noise parameters are given as $\alpha_1=1.85$, $d_0=1$ m, $P_\mathrm{L}(d_0)=59.7$ dB, and $\sigma_\mathrm{dB}=3.2$ dB, which yields $C_0=0.001$. In a free-space scenario, $\alpha_2=2$ and $\sigma_\mathrm{dB}=0$ dB, whilst the remaining parameters are kept identical. The transmitter and receiver noises are configured as $s_j=0.5$ and $z_i=10^{-8}$ ($\forall i,\ j\in\{1,2,3,4\}$), respectively. The external eavesdropper is assumed to be co-located with the transmitter.

The dynamics of each robot are described by \eqref{e1}, where the state represents the $(x,y)$-position coordinates of the robot. The undirected communication graph is given by
\begin{align*}
    \mathcal{A}=\left[\begin{matrix}
        0 & 1 & 0 & 0\\1 & 0 &1& 0 \\0 & 1 &0 & 1\\ 0 & 0 & 1 & 0
    \end{matrix}\right],\quad B=\left[\begin{matrix}
        -1 &1 &0 &0\\0&-1 &1 &0\\ 0 & 0 &-1 &1
    \end{matrix}\right].
\end{align*} For the initial setup, the states of the four robots are assigned as $x_1(0)=\left[\begin{matrix}2 &1\end{matrix}\right]^{\top}$, $x_2(0)=\left[\begin{matrix}7 &3 \end{matrix}\right]^{\top}$, $x_3(0)=\left[\begin{matrix}3 &8 \end{matrix}\right]^{\top}$ and $x_4(0)=\left[\begin{matrix}4 &6 \end{matrix}\right]^{\top}$, with the prescribed relative formations being $d_{12}=\left[\begin{matrix}5 & 0\end{matrix}\right]^{\top}$, $d_{23}=\left[\begin{matrix}0 &5\end{matrix}\right]^{\top}$ and $d_{34}=\left[\begin{matrix}-5 &0\end{matrix}\right]^{\top}$. The parameter within \eqref{e10} is established as $\theta_1=1$. Additionally, the rolling horizon length is chosen to be $T=10$, and the weighting matrices are given by:
\begin{align*}
    Q_i=\left[\begin{matrix}
       q_i &0 \\0 &q_i 
    \end{matrix}\right],  {R}_i=\left[\begin{matrix}
       r_i &0 \\0 &r_i
    \end{matrix}\right],
\end{align*} where $q_i \in \mathcal{Q}$ and $r_i\in\mathcal{R}$, $i=1,2,3,4$.

With the system configurations and parameters established above, the subsequent subsections present the specific numerical simulations, {wherein $t$ denotes the system \eqref{e1} running time, and the eavesdropper is configured to continuously monitor the system from $t = 0$.}
\begin{figure}[htp]
    \centering
    \includegraphics[width=1\linewidth]{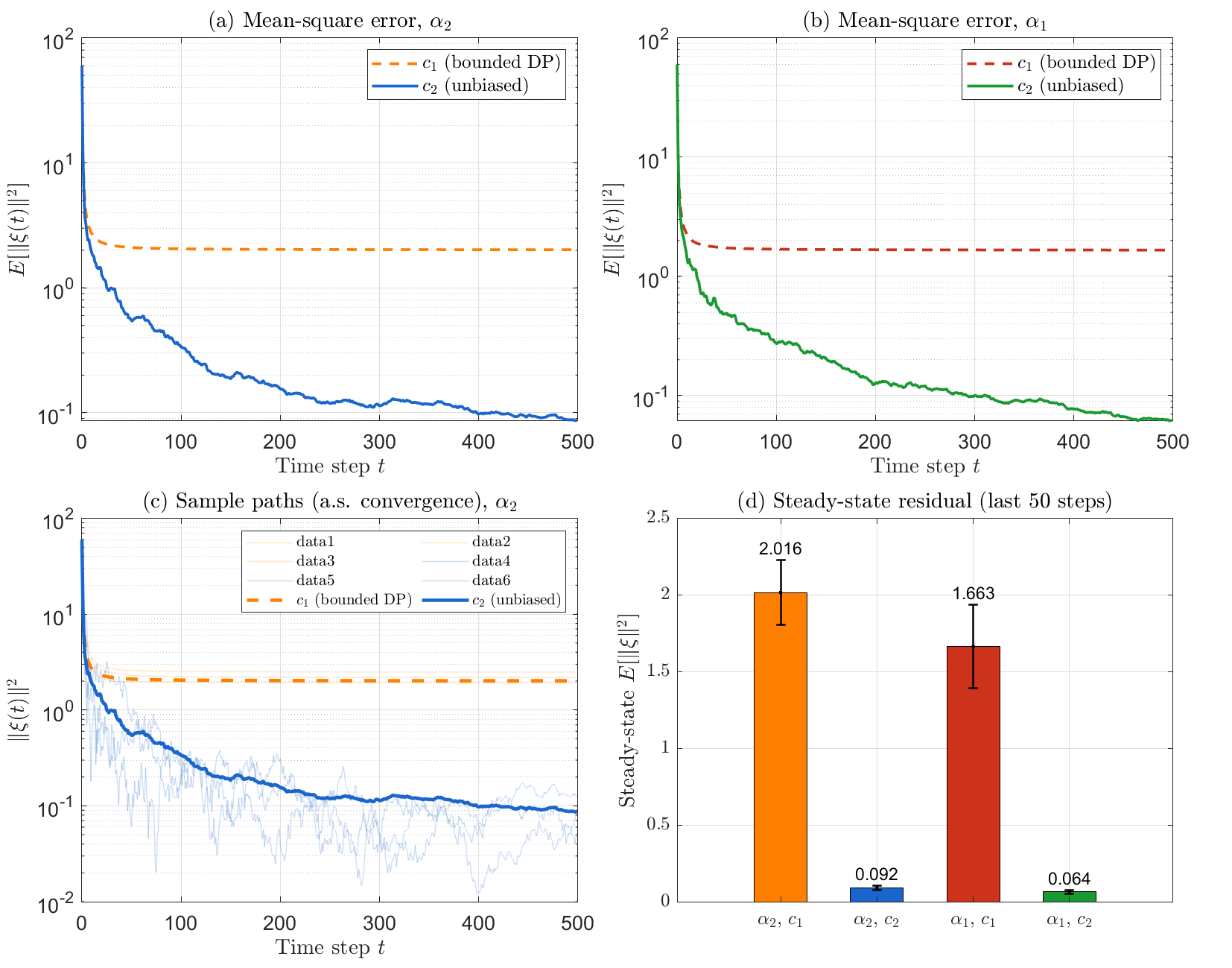}
    \caption{Formation tracking performance of the MAS \eqref{e1} under Algorithm \ref{alg:dpcc}}
    \label{fig2}
\end{figure}

\begin{figure}[htp]
    \centering
    \includegraphics[width=1\linewidth]{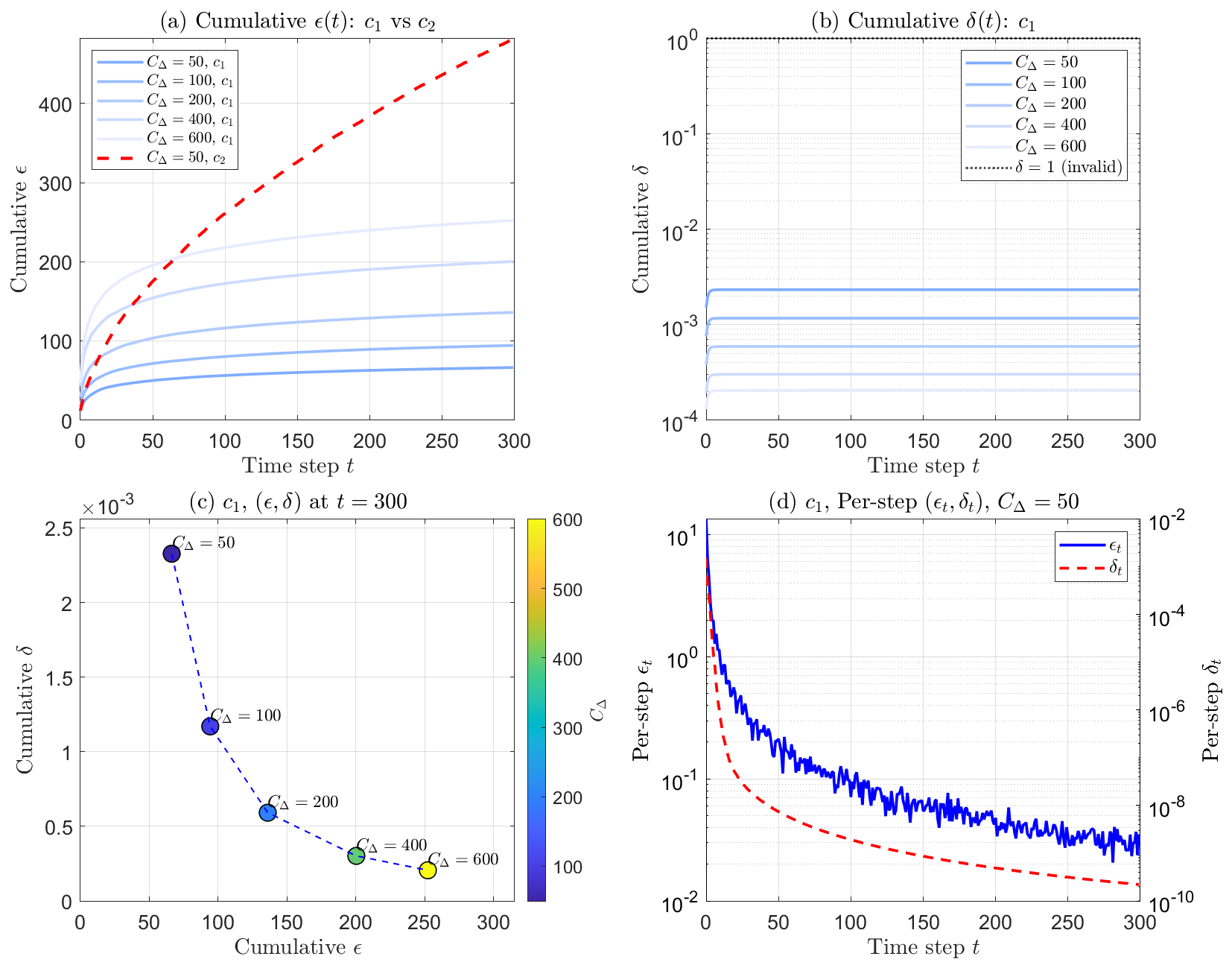}
    \caption{Temporal evolution of the $(\epsilon,\delta)$-DP guarantees under Algorithm \ref{alg:dpcc}}
    \label{fig3}
\end{figure}

\subsection{Convergence and Privacy Performance}

This subsection verifies the proposed algorithm's ability to drive the MAS \eqref{e1} towards the target formation whilst maintaining rigorous DP. The weighting matrices are configured with $q_i = 3$ and $r_i = 1$ for all $i \in \{1,2,3,4\}$. To validate Theorem \ref{th2} and Corollary \ref{co1} with $\alpha_2 = 2$, the penalty factor $c(t)$ is evaluated using two distinct profiles: $c_1(t)=\frac{1}{20(t+1)^{2.2}}\in \ell_1$ and $c_2(t)=\frac{1}{20(t+1)} \in \ell_2, \notin \ell_1$. To substantiate Proposition \ref{pr1}, these configurations are further examined in conjunction with the obstacle scenario where $\alpha_1 = 1.85$. Concurrently, the relaxation probability $\delta'_t=\frac{10^{-5}}{(t+2)^2}$ for DP in Theorem \ref{th3}. It can be verified that $c_1(t)$ and $\delta'_t$ satisfy condition \eqref{con2}. To demonstrate how the parameter $C_{\Delta}$ influences the cumulative privacy guarantees $\epsilon$ and $\delta$, multiple scenarios are conducted by selecting $C_{\Delta} \in \{50, 100, 200, 400, 600\}$. The simulation results are presented in Fig.~\ref{fig2} and Fig.~\ref{fig3}. 

The closed-loop convergence properties, statistically evaluated via Monte Carlo simulations, are delivered in Fig.~\ref{fig2}. As depicted in Figs.~\ref{fig2}(a) and \ref{fig2}(c), under the penalty factors $c_1(t)$ and $c_2(t)$ within a free-space environment $\alpha_2$, the system achieves the formation with a bounded tracking error. The precise values of this error are detailed in Fig.~\ref{fig2}(d), where the height of each bar represents the mean value, whilst the error bars denote the $95\%$ confidence interval. Specifically, to illustrate the almost-sure convergence in Fig.~\ref{fig2}(c), the faint background curves depict the individual stochastic sample paths across multiple realisations, whilst the dark solid lines represent the empirical mean trajectory. Conversely, employing the non-$\ell_1$ factor $c_2(t)$ ensures that the tracking error $\|\xi(t)\|^2$ diminishes to a markedly lower level (0.092 and 0.064 across the respective communication environments), thereby achieving practically unbiased convergence. Furthermore, Fig.~\ref{fig2}(b) demonstrates that the system also sustains convergence under the obstacle-induced noise $\alpha_1$, accompanied by a reduced steady-state residual. These results explicitly corroborate Theorem \ref{th2}, Corollary \ref{co1}, and Proposition \ref{pr1}, respectively.

Fig.~\ref{fig3} presents the corresponding DP guarantees within the feasible sets $\mathcal{Q}=\{2, 2.5,3\}$ and $\mathcal{R}=\{1,1.5,2\}$, subject to $\alpha_1$. Fig.~\ref{fig3}(a) shows a clear contrast: under $c_1(t)$, the cumulative privacy budget $\epsilon(t)$ remains within a finite limit, validating Theorem \ref{th3}. In contrast, $c_2(t)$ causes $\epsilon(t)$ to grow indefinitely. This confirms the fundamental control-privacy dilemma: achieving exact, unbiased convergence prevents the system from maintaining bounded DP guarantees over an infinite horizon. Figs.~\ref{fig3}(a)-(c) show the trade-off related to the operational bound $C_{\Delta}$. Additionally, Fig.~\ref{fig3}(d) reveals that communication noise can provide reasonable per-step DP ($\epsilon_t<10,\delta_t\leq 10^{-4}$ and strictly effective when $\epsilon_t<1$ \cite{near2025guidelines}) for weighting ratios, and that substantial privacy leakage occurs exclusively during the initial transient stage. This suggests that relying on globally designed parameters is overly conservative. Instead, local time-varying regulations of $c(t)$ and $C_{\Delta}$ can be strategically employed to achieve a more flexible and superior transient privacy-performance trade-off.

\subsection{Privacy across Distinct Targets: Ratios vs. Reference States}

To assess whether inherent communication noise provides universal privacy guarantees, this scenario contrasts its protection of weighting ratios with that of the reference state \cite{ma2026inherent}. For a rigorous comparison, both scenarios are subjected to the identical inherent communication noise environment ($\alpha_2$) and operate under the same penalty factor $c_1(t)$. The temporal evolution of the cumulative privacy budget $\epsilon(t)$ for both targets is delivered in Fig. \ref{fig5}.

\begin{figure}[htp]
    \centering
    \includegraphics[width=1\linewidth]{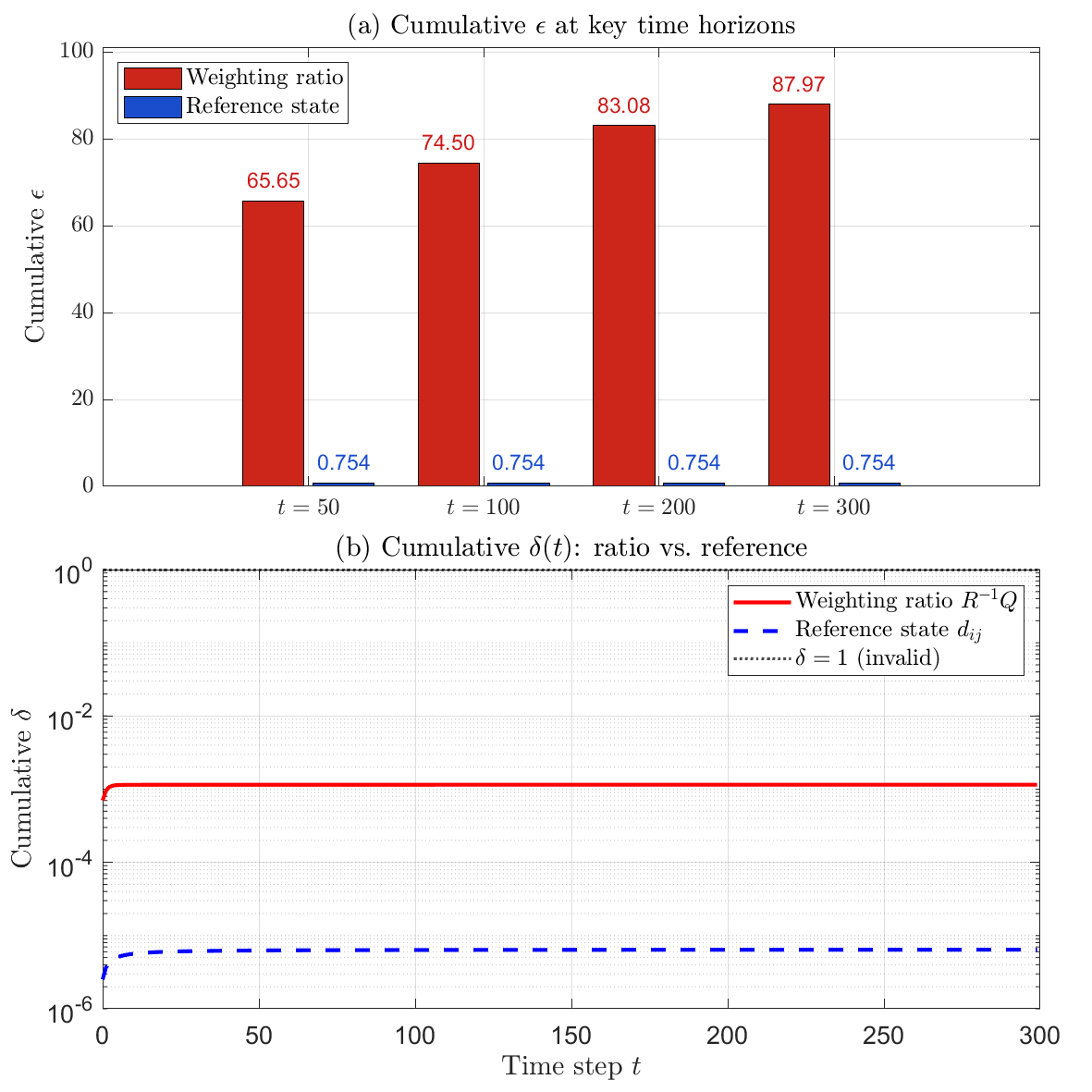}
    \caption{Target-dependent privacy: weighting ratio vs. reference state \cite{ma2026inherent} under noise \eqref{e8}}
    \label{fig5}
\end{figure}

{Fig. \ref{fig5} highlights a clear contrast in privacy protection between the two targets. The weighting ratio's $\epsilon(t)$ grows rapidly early on $t\in[0,100]$ before its rate diminishes. While this deceleration reflects effective privacy protection during the latter stages, the overall privacy guarantee is fundamentally compromised in practice, owing to the massive privacy leakage during the early stages of system operation. In contrast, the reference state maintains a low and nearly constant cumulative privacy leakage across all time horizons, demonstrating robust and enduring privacy protection under continuous eavesdropping.}  This quantitative result demonstrates that whilst inherent communication noise \eqref{e8} serves as a natural privacy mechanism, its protection is not uniform. Rather, its efficacy is target-dependent, leading to distinct rates of privacy loss based on the specific structural role of the protected information within the control mechanism.

\section{Conclusion}

This paper investigated the capacity of inherent communication noise to provide DP in distributed cooperative control MAS. Through the design of a distributed finite-horizon LQR framework, it is established that the cooperative tracking error remains bounded in expectation and convergent in both the mean-square and almost-sure senses for communication environments with path loss exponents $\alpha \leq 2$. Concurrently, the algorithm provides bounded $(\epsilon, \delta)$-DP guarantees for the protected weighting ratios without relying on artificial noise injection. The mathematical derivations also capture a fundamental dilemma: exact asymptotically unbiased convergence and an infinite-horizon bounded DP guarantee cannot be achieved simultaneously. From the perspective of the set-theoretic sensitivity analysis, it is proved that adjacency definitions premised upon weighting ratios yield less conservative privacy bounds in comparison to gradient-based formulations. This result justifies the rationale for the proposed ratio-based adjacency formulation. Numerical simulations are conducted to validate the effectiveness of the proposed algorithm. Furthermore, a comparative simulation with \cite{ma2026inherent} demonstrated that the privacy-preserving efficacy of inherent noise is highly target-dependent, yielding significantly different protection levels based on the structural role of the protected information within the control mechanism. Future research will explore state-dependent adaptive penalty factors to address super-linear noise amplification and secure convergence in severe fading environments where $\alpha > 2$.

\appendix
\subsection{Proof of Lemma 5}
    The expectation is derived using the law of iterated expectations:
    \begin{align*}
        \mathbb{E}[x] = \mathbb{E}_y[\mathbb{E}_{x|y}[x \mid y]] = \mathbb{E}_y[0] = 0.
    \end{align*}
    For the covariance, the law of total covariance is utilised:
    \begin{align*} 
        \mathrm{Cov}[x] = \mathbb{E}_y[\mathrm{Cov}_{x|y}[x \mid y]] + \mathrm{Cov}_y[\mathbb{E}_{x|y}[x \mid y]].
    \end{align*}
    Since $\mathbb{E}_{x|y}[x \mid y] = 0$, the second term on the right-hand side vanishes. For the first term, substitute the given conditional covariance:
    \begin{align*}
        \mathbb{E}_y[\mathrm{Cov}_{x|y}[x \mid y]] = \mathbb{E}_y\left[\left(\frac{a}{y^2} + c\right)\mathbf{I}\right] = \left(a \mathbb{E}[y^{-2}] + c\right)\mathbf{I}.
    \end{align*}
    Given that $\ln(y) \sim \mathcal{N}(0, b^2)$, the variable $y$ follows a log-normal distribution. Recall that the $k$-th moment of a log-normal variable $L \sim \text{LogNormal}(\mu, \sigma^2)$ is given by $\mathbb{E}[L^k] = \exp(k\mu + \frac{1}{2}k^2\sigma^2)$. Setting $\mu=0$, $\sigma=b$, and $k=-2$, obtain:
    \begin{align*}
        \mathbb{E}[y^{-2}] = \exp\left(0 + \frac{1}{2}(-2)^2 b^2\right) = e^{2b^2}.
    \end{align*}
    Substituting this back yields $\mathrm{Cov}[x] = (a e^{2b^2} + c)\mathbf{I}$. The proof is completed. \IEEEQED
\bibliographystyle{IEEEtran}
\bibliography{main}
    \begin{IEEEbiography}[{\includegraphics[width=1in,height=1.25in,clip,keepaspectratio]{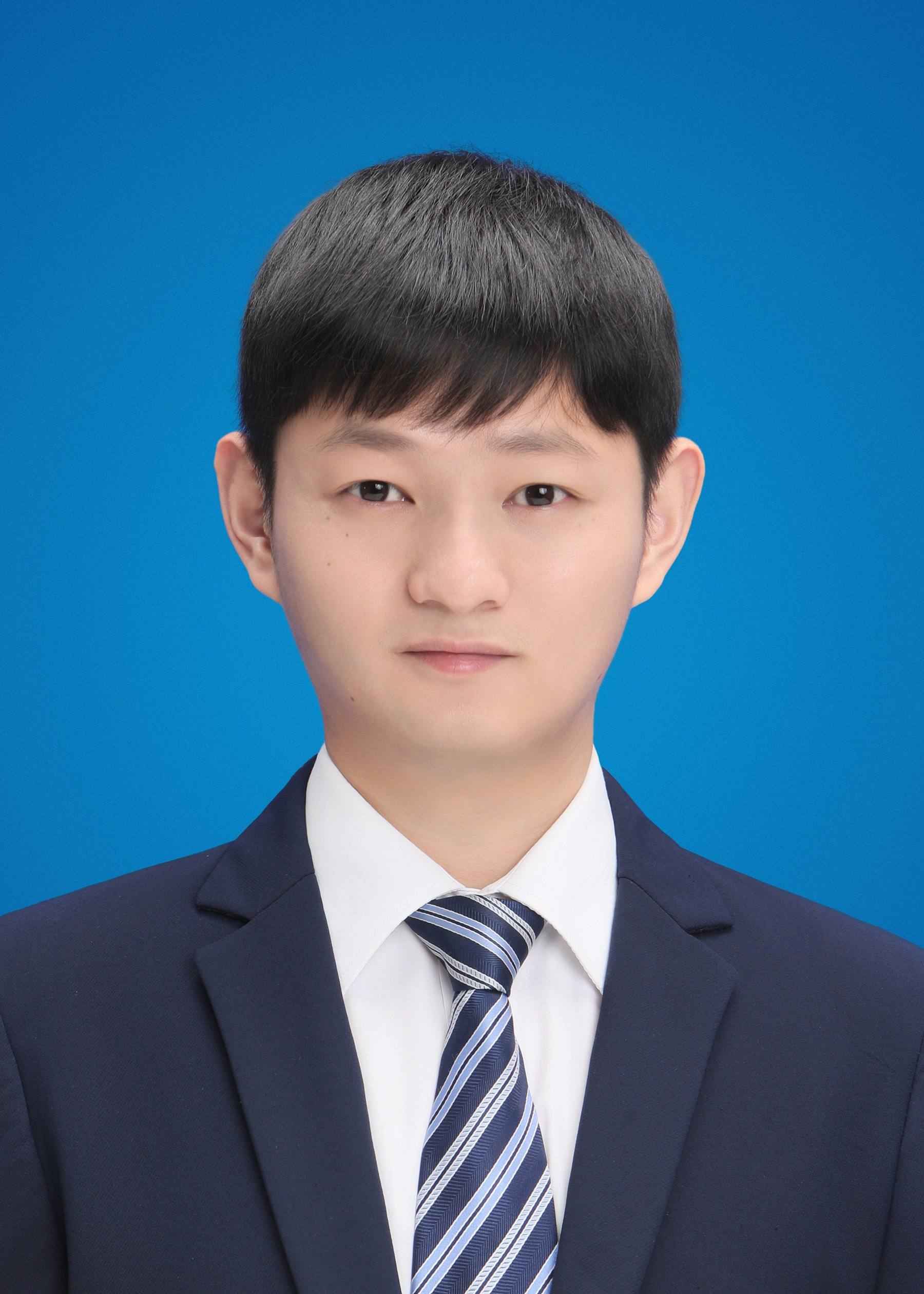}}]{Yuwen Ma}(Student Member, IEEE)
received the B.Eng. degree in automation from Beihang University, Beijing, China, in 2021, and the MSc degree in control engineering from Shanghai Jiao Tong University, Shanghai, China, in 2024. During his Master's studies, he was awarded the National Scholarship in 2023. He is currently working towards the Ph.D. degree under the supervision of Prof. Sarah Spurgeon and Dr Boli Chen at University College London, London, U.K., funded by a UKRI EPSRC DTP program. His Master's dissertation received the Outstanding Master's Thesis Award from the Chinese Association of Automation in 2026. His current research interests include differential privacy, optimal control, and multi-agent systems.
\end{IEEEbiography}

\begin{IEEEbiography}[{\includegraphics[width=1in,height=1.25in,clip,keepaspectratio]{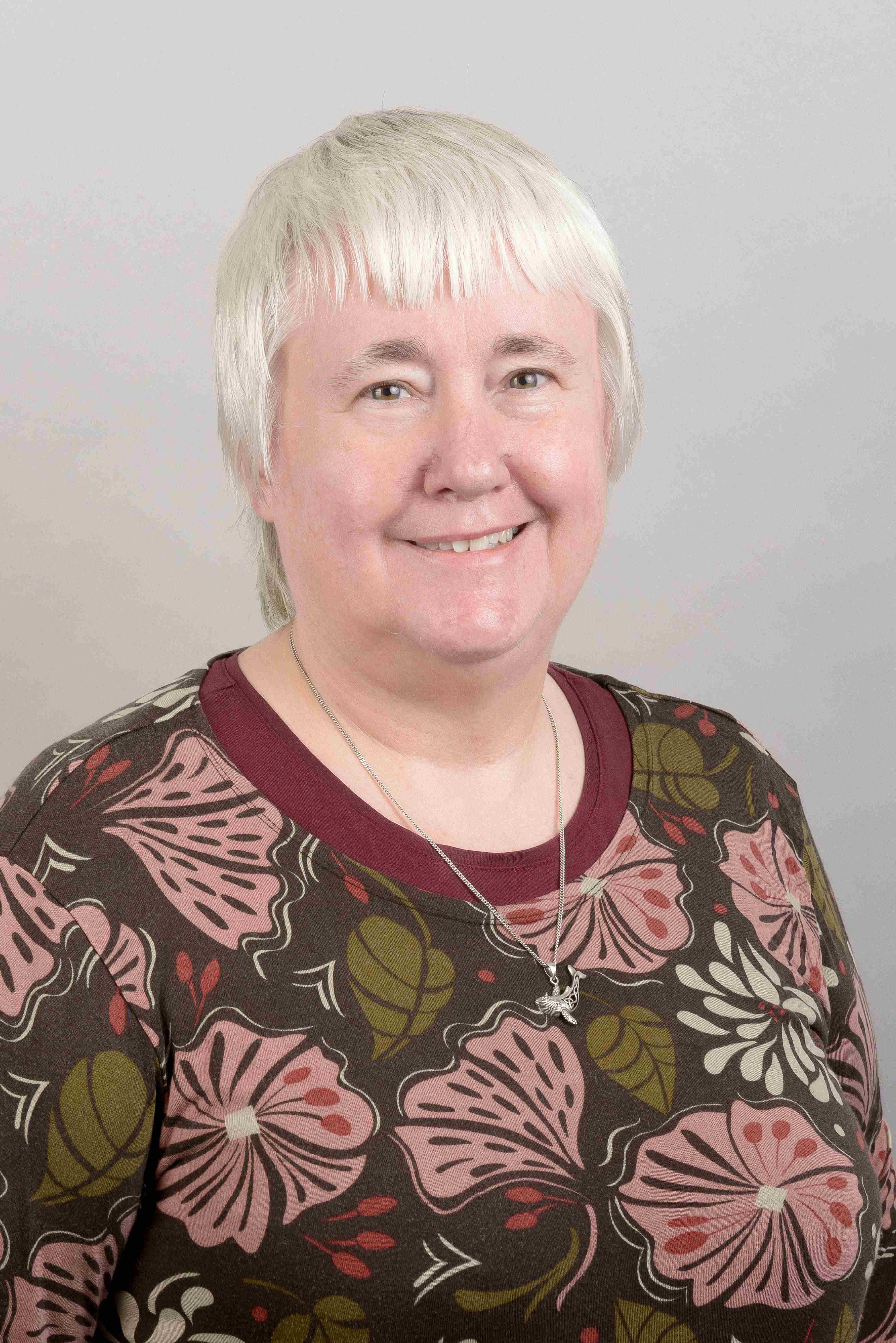}}]{Sarah Spurgeon}(Fellow, IEEE)
received the B.Sc. and D.Phil. degrees from the University of York, York, U.K., in 1985 and 1988, respectively.

She is currently Professor in control engineering and Director of the Robotics Institute at University College London, U.K. Her research interests involve the area of systems modeling and analysis and robust control and estimation. In these areas, she has published more than 300 refereed research articles. She is a Fellow of the Royal Academy of Engineering, InstMC, IET, and IMA. She is the past Editor-in-Chief of IEEE Press and is currently a Director of the IEEE Foundation and the Vice President (Publications) of the International Federation of Automatic Control (IFAC). She received the Honeywell International Medal for distinguished contributions as a Control and Measurement Technologist to developing the theory of control in 2010 and the IEEE Millennium Medal in 2000.
\end{IEEEbiography}

\vspace{-1mm}
\begin{IEEEbiography}[{\includegraphics[width=1in,height=1.25in,clip,keepaspectratio]{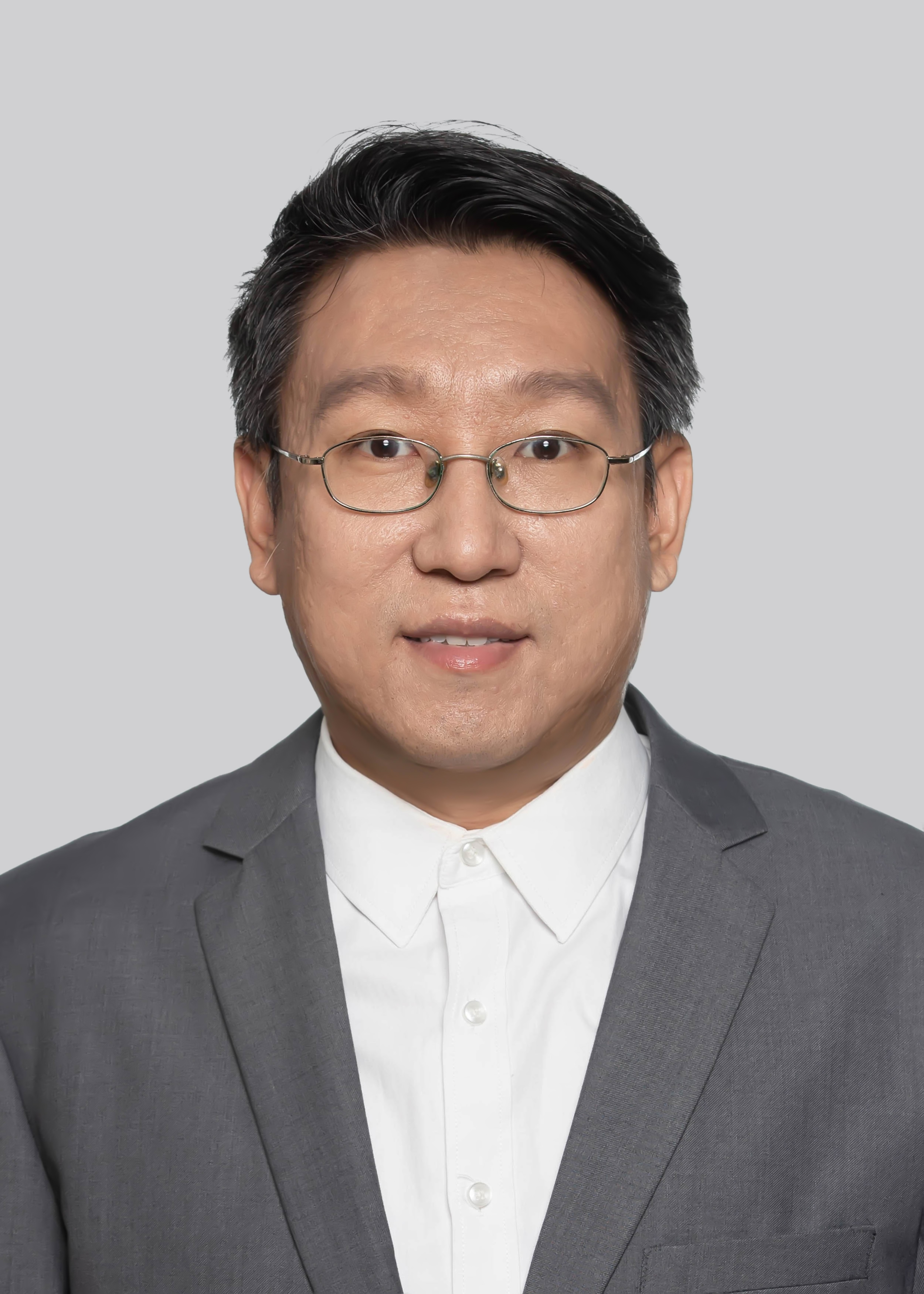}}]{Tao Li}(Senior Member, IEEE) 
received the B.E. degree in automation from Nankai University, Tianjin, China, in 2004, and the Ph.D. degree in systems theory from the Academy of Mathematics and Systems Science, Chinese Academy of Sciences, Beijing, China, in 2009. Since July 2009, he has been with the Academy of Mathematics and Systems Science, Chinese Academy of Sciences, where he is now a Full Professor. His current research interests include stochastic systems, distributed learning, control, and games.

Dr. Li was a recipient of the 28th Zhang Siying Outstanding Youth Paper Award in 2016 and the Best Paper Award of the 7th Asian Control Conference in 2009. He received the 2009 Singapore Millennium Foundation Research Fellowship and the 2010 Australian Endeavour Research Fellowship. He was appointed Dongfang Distinguished Professor by Shanghai Municipality in 2012 and received the Excellent Young Scientists Fund from NSFC in 2015. He was twice elected to the Chang Jiang Scholars Program (Youth Scholar in 2018 and Distinguished Professor in 2023), Ministry of Education, China. He served/serves as Associate Editors for several journals, including IEEE Transactions on Automatic Control, IEEE Control Systems Letters, Systems and Control Letters, Nonlinear Analysis: Hybrid Systems, and Science China Information Sciences.
\end{IEEEbiography}

\begin{IEEEbiography}[{\includegraphics[width=1in,height=1.25in,clip,keepaspectratio]{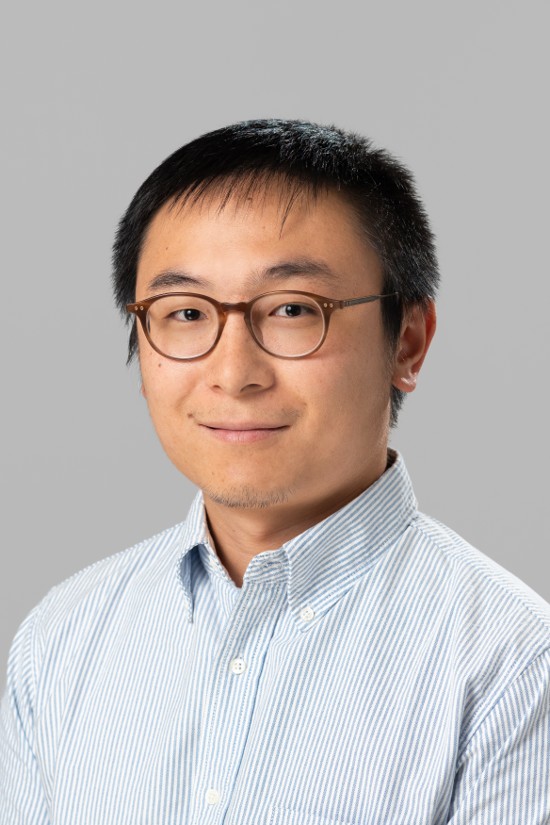}}]{Boli Chen}(M’16–SM’24, IEEE)
received the B.Eng. degree in Electrical and Electronic Engineering from Northumbria University, UK, in 2010. He received the M.Sc. and Ph.D. degrees in Control Systems from Imperial College London, UK, in 2011 and 2015, respectively. He is currently an Associate Professor in the Department of Electronic and Electrical Engineering at University College London (UCL), UK. His research focuses on the control, optimisation, and estimation of complex dynamical systems, with applications to smart cities, including transportation, electric energy systems, and sensor networks. Dr Chen is a member of the IEEE Control Systems Society Technical Committees on “Smart Cities” and “Automotive Controls”. He serves as an Associate Editor for the IEEE Transactions on Intelligent Transportation Systems and the European Journal of Control. He is also a member of the EUCA Conference Editorial Board and the IEEE Intelligent Transportation Systems Conference (ITSC) Editorial Board.
\end{IEEEbiography}
   

\end{document}